\pdfoutput=1
\documentclass[sigconf,nonacm]{acmart}
\setcitestyle{numbers,sort&compress}
\usepackage{textcomp}
\usepackage{xspace}

\newcommand{\mypara}[1]{\par\addvspace{\smallskipamount}\noindent\textbf{#1.}}
\newcommand{\temph}[1]{\emph{#1}}
\newcommand{\remove}[1]{}

\newif\iffull
\fulltrue
\newcommand{\full}[2]{\iffull#1\else#2\fi}

\newcommand{\calV}{\mathcal{V}}
\newcommand{\calG}{\mathcal{G}}
\newcommand{\calT}{\mathcal{T}}

\newcommand{\ia}{\textit{i}}
\newcommand{\ib}{\textit{ii}}
\newcommand{\ic}{\textit{iii}}

\theoremstyle{acmdefinition}

\theoremstyle{acmplain}
\newtheorem*{theoremrestated}{Theorem~\ref{thm:soundness} (Soundness, restated)}
\newtheorem*{corollaryrestated}{Corollary~\ref{cor:run} (restated)}

\title[simulating people]{sGLP:  Simulating People with Stochastic Agents}

\author{Ehud Shapiro}
\affiliation{%
  \institution{London School of Economics and Weizmann Institute of Science}
  \country{UK and Israel }
}

\begin{document}

\begin{abstract}
Grassroots platforms are designed to be owned, operated and governed by their participants, and thus are constituted from person-to-person smartphone apps that are run by the participants.
A grassroots app is specified using a grassroots social contract, which is compiled through an abstraction cascade into a volition-guarded Grassroots Logic Program  (vGLP), a multiagent concurrent logic programming language. A vGLP program in turn specifies a smartphone UI via which people express their volitions and thus control the app.  Prudence dictates large-scale evaluation and testing of grassroots platforms prior to deployment, but using the many people needed to do that is a contradiction in terms. Hence another solution is needed.

Here, we present stochastic GLP (sGLP), designed to allow grassroots platforms to be simulated at scale from the vGLP programs that deploy them, with stochastic agents in place of people.  sGLP is GLP with a rate extension for specifying stochastic behaviour.  A stochastic agent is an sGLP program replacing both the vGLP UI and the person behind it.  It is typed by the dual of the interactive type of the UI and specifies a stochastic profile of a person's behaviour.

We demonstrate sGLP on the grassroots social graph, with a population of men and women drawn from four stochastic profiles --- homophile or indifferent, wary or sociable --- and on grassroots currencies, with four profiles --- spenders or savers, cautious or generous --- and report what the mix of profiles does to the social graph that forms and to the currencies that circulate.

We prove that every run of a vGLP program in which 
sGLP agents replace the UI and the people behind it is a vGLP run.  We implemented sGLP in GLP by Gillespie's method, and ran reproducible simulations with it. 

\end{abstract}

\ccsdesc{Theory of computation~Operational semantics}
\ccsdesc{Software and its engineering~Constraint and logic languages}
\ccsdesc{Software and its engineering~Concurrent programming languages}
\ccsdesc{Computing methodologies~Discrete-event simulation}
\ccsdesc{Computing methodologies~Agent / discrete models}
\ccsdesc{Mathematics of computing~Markov processes}

\keywords{Stochastic Simulation, Gillespie Algorithm, Agent-Based Simulation, Operational Semantics, Concurrent Logic Programming, GLP, Volitional Programming, Grassroots Platforms}

\maketitle

\section{Introduction}\label{sec:introduction}

Grassroots platforms~\cite{shapiro2023grassrootsBA,shapiro2025characterising,shapiro2024grassroots} are designed to be owned, operated and governed by their participants, and thus are constituted from person-to-person smartphone apps that are run by the participants.  A grassroots app is specified using a grassroots social contract~\cite{shapiro2026formalising}, which is compiled through an abstraction cascade~\cite{shapiro2026formalising,shapiro2026volition} into a volition-guarded Grassroots Logic Program (vGLP)~\cite{shapiro2026volition}, a multiagent concurrent logic programming language~\cite{shapiro1989family,shapiro2025glp,shapiro2026types}.  A vGLP program in turn specifies a smartphone UI via which people express their volitions~\cite{lewis2026volitional} and thus control the app.  In vGLP a procedure may put a question to its person, and a clause that reads the answer reduces only after the person has given it, so the transitions of the platform are guarded by the volitions of its people~\cite{lewis2026volitional}.   Each question has an interactive type, and the UI is derived from these types, a widget per question~\cite{shapiro2026volition}. The program is compiled onto GLP~\cite{shapiro2025glp,shapiro2026implementing} and deployed as a mini-app of the grassroots super-app~\cite{shapiro2026gsg}, on iPhone and Android~\cite{dart2024}.  

What such a platform does at scale --- whether a grassroots social graph~\cite{shapiro2023gsn} coalesces from separate groups into one, whether grassroots currencies~\cite{shapiro2024gc,shapiro2026bonds} circulate --- depends on what the people using it do, and cannot be found out by running the app with a handful of them.  Prudence dictates large-scale evaluation and testing of grassroots platforms prior to deployment, but using the many people needed to do that  is a contradiction in terms.  Hence another solution is needed.

Here, we present stochastic GLP (sGLP), designed to allow grassroots platforms to be simulated at scale from the vGLP programs that deploy them, with stochastic agents in place of people.  sGLP is GLP with a rate extension for specifying stochastic behaviour.  A rated goal is reduced at a time drawn, when it is spawned, from the exponential distribution with its rate; the machine runs to rest between two such times, so simulated time advances only with them; and the law of a run is that of Gillespie's first reaction method~\cite{gillespie1977exact}, which we prove equivalent here to the direct method, since a spawned goal is never withdrawn (Section~\ref{sec:sglp}).  A stochastic agent is an sGLP program replacing both the vGLP UI and the person behind it, typed by the dual of the interactive type specifying the UI, with a stochastic profile of a person's behaviour.  A profile is such a program, named, and a population is a mix of profiles from which each agent draws its person (Section~\ref{sec:simulation}).

We demonstrate sGLP on the grassroots social graph~\cite{shapiro2023gsn}, with a population of men and women drawn from four stochastic profiles --- homophile or indifferent, wary or sociable --- and on grassroots currencies~\cite{shapiro2024gc,shapiro2026bonds}, with four profiles --- spenders or savers, cautious or generous --- and report what the mix of profiles does to the social graph that forms and to the currencies that circulate (Sections~\ref{sec:example} and~\ref{sec:currencies}).  The social graph is the platform on which the other grassroots platforms are built, and the currencies are coins among friends in it.  Each demonstration runs 100 agents for a simulated year, over nine mixes of profiles and twenty seeds, and the social graph runs at 1,000 agents for five simulated years.

We prove that every run of a vGLP program in which  sGLP agents replace the UI and the people behind it is a vGLP run.  So what a simulation reports is something the deployed platform can do (Section~\ref{sec:simulation}).  We implemented sGLP in GLP by Gillespie's method~\cite{gillespie1977exact}, and ran reproducible simulations with it.  The implementation is as the stochastic pi-calculus~\cite{priami1995stochastic} was implemented in Flat Concurrent Prolog~\cite{regev2000representation,priami2001application,silverman1988logix}: a transformation of sGLP programs into GLP, and a monitor written in GLP that draws the next event by Gillespie's method over a guard that succeeds when the machine rests; every draw is from the run's seed, so a run is reproducible (Section~\ref{sec:implementation}).  Section~\ref{sec:related-work} reviews related work and Section~\ref{sec:conclusion} concludes.

\section{sGLP: Stochastic GLP}\label{sec:sglp}

A stochastic agent is a program that replaces the UI of a vGLP program and the person behind it: it holds the person's end of each question and answers it, as the widget does on the phone.  sGLP is the language such a program is written in: GLP~\cite{shapiro2025glp,shapiro2026types} with a rate on a goal, which runs in simulated time.  This section defines the language; Section~\ref{sec:simulation} puts a person written in it in the place of the person of a vGLP program.

\subsection{Syntax}

An sGLP program is a typed GLP program~\cite{shapiro2026types} with a rate on a goal.  Its draws are GLP's: \verb|random(|$S?$\verb|, |$N?$\verb|, |$K$\verb|, |$S_1$\verb|)| assigns $K$ an integer uniform from $1$ to $N$, drawn from the seed $S$, and $S_1$ the next seed~\cite{shapiro2025glp}.

{\small\begin{verbatim}
<rated_goal> ::= <goal> @ <rate>
<rate>       ::= <positive_real> / <time_unit>
<time_unit>  ::= second | minute | hour | day | week | year
\end{verbatim}}

\begin{definition}[Rated Goal]\label{def:rated}
A \temph{rated goal} $A$ \verb|@| $r$ in the body of a clause is a goal $A$ with a \temph{rate} $r$, a positive real per unit of time.
\end{definition}

A rated goal is spawned as any body goal is, and is reduced only once its time has come, which is drawn when it is spawned from the exponential distribution with its rate.

A choice among alternatives with given probabilities is a draw and a guard on it in each clause, as the profiles of Sections~\ref{sec:example} and~\ref{sec:currencies} show.

\subsection{Operational Semantics}

Throughout, $N$ is an sGLP program, $P\subset\Pi$ a finite set of agents, and $c$ a configuration of the maGLP transition system over $P$ and $N$~\cite{shapiro2025glp} (Definition~\ref{def:rc-glp}), whose transitions are Reduce, Communicate and Cold-call.  sGLP adds to a configuration a clock and the activation times of its rated goals, and to the transitions one that advances the clock.

\begin{definition}[Configuration, Pending, Quiescent]\label{def:sglp-config}
A \temph{configuration} of sGLP over $P$ and $N$ is a triple $(c,t,\alpha)$: $c$ a configuration of the maGLP transition system over $P$ and $N$, $t\ge0$ the \temph{simulated time}, and $\alpha$ a partial function from the goals of $c$ to \temph{activation times} $\alpha(A)\ge t$, the goals in its domain being \temph{pending}.  It is \temph{quiescent} if no Communicate, no Cold-call and no Reduce of a goal that is not pending is enabled in $c$.
\end{definition}

\begin{definition}[sGLP Transition System]\label{def:sglp-ts}
The \temph{sGLP transition system} over $P$, $N$ and a configuration $c^0$ of the maGLP transition system over $P$ and $N$ has the configurations of Definition~\ref{def:sglp-config}, the initial configuration $(c^0,0,\emptyset)$, and the transitions $(c,t,\alpha)\rightarrow(c',t',\alpha')$:
\begin{enumerate}
\item \textbf{Machine:} $c\rightarrow c'$ is a Communicate or Cold-call transition, or a Reduce of a goal $A$ that is not pending; $t'=t$; and $\alpha'$ is $\alpha$ with $\alpha'(B)=t+\tau$ for each rated goal $B$ \verb|@| $r$ in the body the Reduce spawns, $\tau$ drawn from the exponential distribution with rate $r$, the draws independent.
\item \textbf{Release:} $(c,t,\alpha)$ is quiescent, $A$ is a pending goal with $\alpha(A)$ least, $c'=c$, $t'=\alpha(A)$, and $\alpha'$ is $\alpha$ without $A$.
\end{enumerate}
\end{definition}

\begin{definition}[sGLP Run]\label{def:sglp-run}
An \temph{sGLP run} of $N$ over $P$ from $c^0$ is a sequence of configurations of the sGLP transition system over $P$, $N$ and $c^0$, the first initial and each a transition of the one before, such that its configurations of the maGLP transition system, with each repeated one taken once, form a proper safe computation of that system from $c^0$ (Definitions~\ref{def:rc-ts} and~\ref{def:rc-proper}, each configuration $(G,\sigma)$ read as $(\emptyset,G,\sigma)$); an \temph{sGLP run} of $N$ over $P$ is one from the initial configuration of the maGLP transition system over $P$ and $N$.
\end{definition}

Simulated time advances only at a Release: a machine step takes no time, and the clock jumps to the earliest activation time once the machine has nothing to do at the present time.  A released goal is an ordinary goal: where its input has not arrived it suspends and reduces when it does, with no further delay.

\begin{proposition}[Time to the Next Release]\label{prop:first-reaction}
Let $(c,t,\alpha)$ be a quiescent configuration of an sGLP run with pending goals $A_1,\ldots,A_m$ of rates $r_1,\ldots,r_m$, and $R=\sum_i r_i$.  Conditional on the run up to it, the next Release comes after a time exponentially distributed with rate $R$, and releases $A_i$ with probability $r_i/R$, independently of that time.
\end{proposition}
\begin{proof}
$\alpha(A_i)=t_i+\tau_i$, with $t_i\le t$ the time $A_i$ was spawned and $\tau_i$ exponential with rate $r_i$, the $\tau_i$ independent of one another and of every other draw of the run.  The run up to $(c,t,\alpha)$ depends on $\tau_i$ only through $\alpha(A_i)\ge t$, so conditional on it $\alpha(A_i)-t$ is exponential with rate $r_i$, by memorylessness, and the $\alpha(A_i)-t$ are independent.  The minimum of independent exponential variates with rates $r_1,\ldots,r_m$ is exponential with rate $R$, and is the $i$-th with probability $r_i/R$, independently of its value.
\end{proof}

This is Gillespie's first reaction method~\cite{gillespie1977exact}, and the proposition is its equivalence to the direct method, which holds here because a spawned goal is never withdrawn.

\section{Simulating a vGLP Program}\label{sec:simulation}

On the phone, the person's end of each interactive variable is held by a construct process, which has the question's widget drawn and assigns the person's answer; the widget's type is the dual of the interactive type, writing where the program reads and reading where the program writes~\cite{shapiro2026volition}.  In a simulation the person's end is held by a relay, which hands a copy of the question to a goal of a person procedure of the dual type and passes the goal's answer on to the program, and the population is a mix of profiles of person, each profile a program.  The vGLP program, its types and its clauses are those the phone runs.

{\small\begin{verbatim}
<person_declaration> ::= <type> =::= <procedure_name> .
<profile>     ::= person <name> .
                  <person_declarations> <program>
<run>         ::= run <integer> agents [ <mix_list> ]
                  until <time> seed <integer> .
<mix_list>    ::= <mix> | <mix> , <mix_list>
<mix>         ::= <dimension> ~ ( <name> : <probability> ; ... )
\end{verbatim}}

The definitions and results of vGLP~\cite{shapiro2026volition} that this section and Theorem~\ref{thm:soundness} rest on --- interactive types, the vmaGLP transition system, the canonical compilation and its mapping back to the source --- are recalled in Appendix~\ref{app:recalled}.  Throughout, $M$ is a vGLP program, $\calT(M)$ the set of its interactive types, each in writer or reader mode, and $\lceil M\rceil$ its canonical compilation (Definition~\ref{def:rc-compilation}), which gives the interactive types distinct functors, $A$ being the type of its asks.

\begin{definition}[Dual, Person Procedure, Person Declaration]\label{def:dual}
The \temph{dual} of an interactive type $T$ is $T^\circ := T'?$ where $T = T'$ and $T^\circ := T'$ where $T = T'?$; a \temph{person procedure} of $T$ in an sGLP program $N$ is a procedure $p$ of $N$ declared \verb|procedure| $p(T^\circ$\verb|, Integer?)|; and a \temph{person declaration} is a pair $(T, p)$, written $T$ \verb|=::=| $p$, of an interactive type $T \in \calT(M)$ and a person procedure $p$ of $T$.
\end{definition}

The person's end of an interactive variable is of the dual type, and a person procedure takes a copy of it, of the same type, with a seed.

\begin{definition}[Profile, Dimension]\label{def:profile}
A \temph{profile} over $M$ is a triple $\phi = (\mathit{name}, N, \Delta)$ of a name, an sGLP program $N$ and a finite set $\Delta$ of person declarations with distinct interactive types, each naming a person procedure of $N$, $\calT(\phi)$ being the set of those types; a \temph{dimension} over $M$ is a pair $(d, \Phi)$ of a name and a finite set $\Phi$ of profiles over $M$ with one set of interactive types, $\calT(\Phi)$; and dimensions $(d_1,\Phi_1),\ldots,(d_k,\Phi_k)$ over $M$ are \temph{disjoint and covering} if $\calT(\Phi_1),\ldots,\calT(\Phi_k)$ are pairwise disjoint with union $\calT(M)$.
\end{definition}

\begin{definition}[Stochastic Person, Population]\label{def:sglp-population}
A \temph{stochastic person} over disjoint and covering dimensions $(d_1,\Phi_1),\ldots,(d_k,\Phi_k)$ is a tuple $(\phi_1,\ldots,\phi_k)$ with $\phi_j \in \Phi_j$ for $j = 1,\ldots,k$, and a \temph{population} over $M$ is a tuple $(n, D, \mu, s)$ of an integer $n \ge 1$, disjoint and covering dimensions $D = ((d_1,\Phi_1),\ldots,(d_k,\Phi_k))$ over $M$, a tuple $\mu = (\mu_1,\ldots,\mu_k)$ of probability distributions, $\mu_j$ on $\Phi_j$, and an integer $s$, its \temph{seed}; its \temph{agents} are $P = \{1,\ldots,n\}$.
\end{definition}

A stochastic person therefore declares each interactive type of $M$ exactly once, by one of its profiles.

\begin{definition}[Well-Formed]\label{def:wellformed}
A population is \temph{well formed} if each person procedure its profiles declare is declared with the type its declaration requires, the predicates of the profiles' programs are distinct from those of $\lceil M\rceil$, and $\mu_j$ sums to one over $\Phi_j$ for each $j$.
\end{definition}

The transformation of Section~\ref{sec:implementation} refuses a population that is not well formed.  A run declaration draws from the population's seed, for each agent, a stochastic person --- one profile per dimension, by that dimension's distribution --- and a seed; Section~\ref{sec:implementation} says how.  For the rest of the section, $(n, D, \mu, s)$ is a well-formed population over $M$ with agents $P$ and, for each $a \in P$, $\pi(a)$ is a stochastic person over $D$, $s_a$ an integer and $g_a$ a goal of $M$, its initial goal, with $g'_a$ the goal $g_a$ extended with an ask stream $U_a$ and with its calls of volitional procedures replaced, as in $\lceil M\rceil$.

\begin{definition}[Copy, Relay]\label{def:relay}
For $T \in \calT(M)$, $Q$ the type of the asks' second arguments, the union over the functors $\lceil M\rceil$ gives the interactive types of $M$, $t$ the one it gives $T$ (Definition~\ref{def:rc-compilation}), and $\mathit{relay}_T$ a procedure declared \verb|procedure| $\mathit{relay}_T(Q?$\verb|, |$T$\verb|)|, a goal $\mathit{relay}_T(R, Y)$, with $R$ the term $t(E)$ of an ask, or a reader assigned it, the ask \verb|ask(|$T$\verb|, |$t$\verb|(|$E$\verb|))|, $E$ the person's end of its interactive variable, and $Y$ one end of a fresh variable pair, the \temph{copy}, whose other end is of type $T^\circ$, is a \temph{relay} if each reduction of it and of the goals descended from it (\ia) assigns a variable \temph{received} through $R$ --- a variable of $t(E)$ or of a term assigned to a variable so received, or the pair of one --- only where it is a writer at a position the person writes, and then either the reader of a fresh writer, which \temph{takes the place} of the writer assigned, or its \temph{answer}, at most one answer: where the writer's place in the copy holds a term ground but for variables at its stream and list positions and at the positions the program writes, a list cell or \verb|[]| where the writer is of a stream or list type, that term with a fresh reader at each stream and list position the person writes, but where the copy ends a list there with \verb|[]|, and a fresh writer at each outermost position the program writes; (\ib) puts no variable received through $R$ in a term it assigns to a variable of the copy --- $Y$, a variable of a term assigned to a variable of the copy, or the pair of one; and (\ic) passes on to the goals it spawns each variable received through $R$ that its goal holds and does not assign, and the writer of each fresh reader and the reader of each fresh writer it assigns a received variable with.
\end{definition}

A relay holds the person's end of a question, shows the copy the program's output, and passes the copy's answer on to the program once it is complete, as a construct process forms an answer from the person's input; an answer of a stream or a list type is passed on one element at a time.  A writer passes from goal to goal of a relay as GLP passes any writer on, by the reader of a fresh writer that takes its place.

\begin{definition}[Person Process]\label{def:person-process}
The \temph{person process} of a stochastic person $\phi$ over $D$ is the procedure $\mathit{person}_\phi$, written \verb|person_phi| and declared \verb|procedure person_phi(Stream(|$A$\verb|)?, Integer?)|, with the clause \verb|person_phi([], _).| and, for each person declaration $(T, p)$ of a profile of $\phi$, the clause
\begin{verbatim}
person_phi([ask(T, R)|As], S) :-
    random(S?, 2147483646, K, S1),
    p(Y?, K?), relay_T(R?, Y), person_phi(As?, S1?).
\end{verbatim}
with \verb|Y| and \verb|Y?| exchanged where $T$ is in reader mode, $\mathit{relay}_T$ being a procedure whose goals $\mathit{relay}_T(R, Y)$ are relays (Definition~\ref{def:relay}), its predicate and those of the procedures it reaches distinct from those of $\lceil M\rceil$ and of the profiles' programs.
\end{definition}

The end the ask carries is the writer where $T$ is in reader mode and the reader where it is in writer mode (Definition~\ref{def:rc-compilation}).  On each ask the person process hands the ask's term, which holds the person's end of the question, to the relay it spawns, and the other end of the copy to the person procedure its profile declares for the question's type, with a seed drawn from its own.

\begin{definition}[Simulation Program]\label{def:simulation-program}
The \temph{simulation program} of $(M, P, \pi)$ is the sGLP program $N$ whose clauses are those of $\lceil M\rceil$ other than the construct processes' and the dispatcher's, those of the programs of the profiles of $\pi(a)$ for every $a \in P$, and those of $\mathit{person}_{\pi(a)}$ for every $a \in P$, with the procedures $\mathit{relay}_T$ they call and the procedures these reach, and whose \temph{initial configuration} $c^0$ has at each $a \in P$ the resolvent $(\{g'_a,\ \mathit{person}_{\pi(a)}(U_a?, s_a)\}, \emptyset)$.
\end{definition}

A \emph{stochastic agent} is an agent of the sGLP transition system over $P$, $N$ and $c^0$: it runs the platform's program as the phone does, with the person process, the relays and the programs of its stochastic person in place of the dispatcher, the construct processes and the person behind them.  A run declaration \verb|run| $n$ \verb|agents [|$\ldots$\verb|] until| $T$ \verb|seed| $s$ runs the simulation program of $M$ over agents $1,\ldots,n$ with the stochastic persons and seeds its population draws from $s$, until the simulated clock reaches $T$; the initial goal of each agent is stated with the program, as Sections~\ref{sec:example} and~\ref{sec:currencies} state theirs.  Every draw of a person is from a seed derived from the population's seed, so the draws of a run are reproducible from that seed, and the run itself is reproducible where the machine's schedule is deterministic.  For the rest of the section, $r = (c^1,t_1,\alpha_1), (c^2,t_2,\alpha_2), \ldots$ is an sGLP run of $N$ over $P$ from $c^0$.

\begin{definition}[Platform Goal, Person Goal, Interactive Variable, Interface Variable, Ask-Stream Variable, Question of a Writer, Open]\label{def:sglp-goals}
A goal in a resolvent of $r$ is a \temph{person goal} if it is the goal $\mathit{person}_{\pi(a)}(U_a?, s_a)$ of $c^0$, for some $a \in P$, or a goal a Reduce of a person goal spawns, and a \temph{platform goal} otherwise; the \temph{interactive variables} of an agent $a$ in $r$ are the variables $X$ of the asks \verb|ask(|$T$\verb|, |$t$\verb|(|$X$\verb|))| written by the reductions of asking clauses at $a$ in $r$; the \temph{interface variables} $I_a$ of $a$ in $r$ form the least set closed under pairs that contains the interactive variables of $a$ and the variables of every term assigned, by a transition $c^i \rightarrow c^{i+1}$ of $r$, to a writer $X \in I_a$ open in $c^i$; the \temph{ask-stream variables} of $a$ are the variables at the ask-stream argument of a platform goal of $a$ and at the arguments of the \verb|merge| goals the compilation adds, in a configuration of $r$, the variables a \verb|merge| clause binds to the elements it moves, and the tail variable of each list assigned to one of them; the \temph{question} of a writer of $I_a$ is the question of the writer whose place it takes (Definition~\ref{def:relay}), where it takes one, and the writer itself otherwise; and a writer is \temph{open} in a configuration if its reader occurs there other than in the assignment by which the writer takes another's place, or the writer whose place it takes is open there.
\end{definition}

\begin{definition}[Log]\label{def:sglp-log}
For $i \ge 1$, let $\theta_i$ be the restriction to $I_a$ of the assignments the transition $c^i \rightarrow c^{i+1}$ adds at $a$ but for those by which a writer takes another's place, each read as an assignment to the question of its writer, the identity substitution where it adds none; the transition is \temph{logged} if it is a Reduce of a person goal of $a$ and $\theta_i$ is not the identity, and its \temph{entries} are the triples $(t_i, a, s\{X := v\})$ for the assignments $X? := v$ of $\theta_i$, $s$ being the term in which $X$ was handed to the person --- the ask's term where $X$ is an interactive variable, the term a platform goal's Reduce at $a$ assigned to a variable of $I_a$ where $X$ is among that term's variables, and $X$ itself where a person goal introduced it --- and every variable of $s\{X := v\}$ replaced by the anonymous variable; and the \temph{log} of $r$ is the sequence of the entries of its logged transitions, in the order of $r$.
\end{definition}

An entry of the log is an answer with the question it answers, as the person was shown it, the positions still open anonymous, so the log is what the people were asked and what they answered and nothing else, one entry per answer, an answer of a stream type one element at a time; a run adds nothing to the program it runs.  Reading the logs and analysing the runs --- counting, aggregating, comparing populations, plotting --- is separate software, outside the language and outside this paper.

A simulated run is a run of the platform.  The mapping that says so reads the platform's configuration off the platform goals, as vGLP's compilation mapping (Definition~\ref{def:rc-sigma}) reads it off the compiled goals, with the person goals in the place of the construct processes and the dispatcher.

\begin{definition}[Simulation Mapping]\label{def:simulation-sigma}
The \temph{simulation mapping} $\sigma$ of $r$ takes a configuration $(c,t,\alpha)$ of $r$, with $\sigma_a$ the readers substitution of $a$ in $c$, to the configuration of the vmaGLP transition system over $P$ and $M$ with, at each agent $a \in P$, the agent state $(H, G, \sigma_r)$:
\begin{itemize}
\item $G$ = the multiset of the platform goals of $a$ in $c$, each read back to its source with its ask stream dropped: $q(S_1,\ldots,S_n,D)$ as the $n$-ary goal $q(S_1,\ldots,S_n)$, $q_1(S_1,\ldots,S_n,X?,D)$, and\\ $q_1(S_1,\ldots,S_n,X,D)$ as the asked goal, with the identifier of the goal its asking clause reduced, and any other platform goal as itself; the \verb|merge| goals the compilation adds and the person goals contributing nothing; and in every goal an asking clause's pair for an argument read as that argument, as in Definition~\ref{def:rc-sigma}, and a writer of $I_a$ and its reader read as its question and its question's reader (Definition~\ref{def:sglp-goals});
\item $H$ = the variables of $I_a$ that occur in a person goal or a \verb|merge| goal the compilation adds of $a$ in $c$, in a term $\sigma_a$ assigns to an ask-stream variable, or in a term of an assignment other than an answer that a Reduce of a person goal added to $\sigma_a$ for a reader a person goal holds, read as in $G$, each with its type, a writer where the person writes and a reader where the program writes; and the assignments to variables of $I_a$ that the Communicate transitions of $r$ up to $(c,t,\alpha)$ delivered to person goals of $a$, shown;
\item $\sigma_r$ = the restriction of $\sigma_a$ to the assignments added by a Reduce of a platform goal of $a$ to variables other than ask-stream variables and the variables of the pairs by which an asking clause passes its arguments on, an asking clause's assignment $W? := S_l?$ of an output position read through as in Definition~\ref{def:rc-sigma}, and to the answers a Reduce of a person goal of $a$ added to writers of $I_a$ open in the configuration it reduced, read as in $G$.
\end{itemize}
\end{definition}

\newcommand{\soundnesssetting}{Let $r$ be an sGLP run of the simulation program $N$ of $(M,P,\pi)$ over $P$ from its initial configuration $c^0$, $\pi$ over a well-formed population, and $\sigma$ the simulation mapping of $r$.}
\newcommand{\soundnessstatement}{$\sigma(r)$ is a proper safe computation of the vmaGLP transition system over $P$ and $M$ from $\sigma(c^0)$ in which every Answer and every Dismiss is the image of a Reduce of a person goal.}
\newcommand{\runstatement}{Where every $g_a$ is the initial goal\\ \texttt{agent(ch(\_?,\_),ch(\_?,\_))} of $M$, $\sigma(r)$ is a proper safe run of the vmaGLP transition system over $P$ and $M$.}
\soundnesssetting
\begin{theorem}[Soundness]\label{thm:soundness}
\soundnessstatement
\end{theorem}
\begin{corollary}\label{cor:run}
\runstatement
\end{corollary}

The proofs are in Appendix~\ref{app:proofs}.  What a simulation reports is therefore something the deployed platform can do, with the people's answers given by the stochastic agents: a person procedure's goal computes the answer into its copy of the question, and the relay assigns the program's variable with it.  Simulated time advances only at a Release, and $\lceil M\rceil$ has no rated goal, so it advances only with the persons' acts.  The converse --- that every run of the platform is the image of a simulated run under some population --- is not claimed: a simulation shows what its profiles do.

\section{Implementation}\label{sec:implementation}

sGLP is implemented in GLP by Gillespie's method~\cite{gillespie1977exact}, as the stochastic pi-calculus was implemented in Flat Concurrent Prolog~\cite{regev2000representation,priami2001application,silverman1988logix}: by a transformation of sGLP programs into GLP, a monitor process written in GLP, and one thing from the machine, a guard that succeeds when it has come to rest.  The engine, the compiler and the type checker of GLP~\cite{shapiro2026implementing} are not changed.

\mypara{The guard}  A machine step carries no simulated time, so the monitor may draw only when no transition of the machine is possible; with no way to learn that, it cannot know when to draw.  GLP has the guard \verb|when_idle|, which suspends while the machine has a Reduce or a Communicate to make and succeeds when it has none~\cite{shapiro2025glp}.  It is sound for one machine and not for many --- in a multiagent system a communication changes an agent's configuration without that agent taking a transition --- which is why the simulator hosts the whole population in one machine.  SpiFcp's monitor did the same, waiting on the Logix system's \verb|idle_wait| before its next draw~\cite{silverman1988logix}.

\mypara{The monitor}  The monitor is Gillespie's algorithm~\cite{gillespie1977exact} written in GLP, in its first reaction form (Section~\ref{sec:sglp}): a GLP goal holding the simulated clock, the pending requests in a queue by activation time, the run's seed and its horizon.  Its input is one stream, a mutual reference of GLP's into which every person goal appends; a request \verb|rated(R, Tok)| asks for a goal of rate \verb|R| to be released, and the monitor assigns it an activation time, the clock plus $-\ln U/R$ with $U = K/(2^{31}-1)$ and $K$ drawn by \verb|random| from its seed, and queues it.  Its release clause is guarded by \verb|when_idle|: when the machine rests, every request sent has arrived, so the queue is that of the configuration; the monitor advances the clock to the least activation time and binds that request's \verb|Tok|, which is Release.  It stops when the least activation time is past the horizon, binding every pending token to \verb|stop|, or when the queue is empty at rest, and writes the clock.  The clock counts simulated seconds.  The module, the implementation of the core of Gillespie's algorithm, is Appendix~\ref{app:monitor}.

\mypara{The transformation}  A profile is translated into GLP: its person procedures take a third argument, the monitor's reference, which the procedures they call that hold rated goals take as well; and a rated goal \verb|G @ r| becomes the request \verb|rated(r, Tok)| appended to that reference, followed by \verb|G| with \verb|Tok?| as an added first argument, matched by a clause that takes the goal only once the monitor has bound the token to \verb|go| and by one that ends it once the monitor has bound it to \verb|stop|, a rated goal's procedure taking inputs only, which the transformation checks; a rated goal that writes an answer receives its writer inside a term it reads.  A person declaration \verb|T =::= p| becomes a clause of the person process, and the run declaration becomes the initial goals of the run.  The transformation is a GLP program: it reads an sGLP program in a term form, one GLP term per declaration and clause, checks that its population is well formed (Section~\ref{sec:simulation}), and prints GLP source, as the compiler of grassroots social contracts reads contracts~\cite{shapiro2026formalising}; the syntax printed in this paper is the surface form of that term form.  The profiles and populations of Sections~\ref{sec:example} and~\ref{sec:currencies} are translated by it; the vGLP program's compilation and the harness of a run are written by hand.  GLP's own \verb|Goal@Agent|, which places a goal at an agent in a boot clause, is untouched.

\mypara{The person process}  In vGLP's canonical compilation an asking clause writes an ask --- the question's interactive type and the person's end of its variable --- on an ask stream that the compiler threads through every procedure that reaches a question, and the dispatcher, reading the initial goal's ask stream, spawns on each ask the construct process that has the bridge draw the widget~\cite{shapiro2026volition}.  In a simulation the dispatcher's place is taken by a GLP process, the \emph{person process} of the agent: it holds the agent's stochastic person, one profile per dimension, and a seed, and on each ask it spawns the goal $p(Y, S, \mathit{Mon})$ of the person procedure $p$ that its profile declares for the type, with $Y$ the person goal's end of a copy of the question, $S$ a seed it draws from its own generator, a value with the full range, and $\mathit{Mon}$ the monitor's reference, and the relay of the question (Definition~\ref{def:relay}), its answer and present procedures, by the question's moded type.  The person goal answers into its copy, and the relay hands the program's variable the answer once the positions the person writes are ground but for those of a stream or a list, which it hands over one element at a time, with fresh writers at the outermost positions the program writes, whose readers it keeps, showing the person goal the program's terms there the same way --- vGLP's Answer, which waits on nothing the program writes --- and that is how the person process comes to log it.  The other way, the relay shows the person goal the question once the positions the program writes are ground, the positions the person writes inside them left open, each answered as a question of its own --- vGLP's Communicate of the program's output to the person.  Every term crossing between the program and the person goal, in either direction and at every depth, crosses the relay so, a question the program puts inside an answer included.  The platform's program compiles as for the phone.

\mypara{The run}  The population runs as one program in one machine, and an agent is an attribution: its initial goal as the program states it, with the agent's number and the lists it is created with, its person process, and what they spawn.  The agents are numbered $1,\ldots,n$.  The harness's goals are the network between the agents, a mutual-write merge of their output streams and a tree of routing goals delivering each message to the input stream of the agent it names, and the monitor.  Every draw of the run is by \verb|random| from the run's seed: the monitor draws from it directly, and the harness draws each agent's profiles and its person process's seed from a separate block of a second generator started from it, so that the agents' draws are independent streams and the run is reproducible where the machine's schedule is deterministic, as GLP's is.

\mypara{The log}  The person process writes the log of its agent: one line for each answer a relay hands the program, with the simulated time, which it asks of the monitor, the agent, the interactive type and the term handed over, the positions the program writes as variables.  Everything the log records crosses a relay, and every answer does.  The harness names the log's sink when the run starts.

The implementation decisions beyond Sections~\ref{sec:sglp} and~\ref{sec:simulation} are in Appendix~\ref{app:notes}.  We claim, informally, that a run of the transformed program, up to the monitor's stop, is an sGLP run of the simulation program: the transformed clauses reduce as the source clauses do, a request's token is bound only by the monitor, only when \verb|when_idle| has succeeded, which is quiescence, and only for the least activation time, which is Release; and the monitor's reductions, the harness's, and the person process's that ask the clock and write the log are stutters.

\mypara{The claim, tested}  Two tests in the repository check the claim where it bites.  With a rated goal and an unrated goal that runs on, the unrated goal runs to rest before the rated one is released, whichever is spawned first: \verb|when_idle| succeeds at quiescence and not before.  With four pending requests of rates 1/second, 2/second, 30/minute and 3600/hour and nothing else, the first 10,000 times between releases have mean $0.2206$ s and variance $0.04934$ s$^2$, against $1/4.5 = 0.2222$ and $1/4.5^2 = 0.04938$ of Proposition~\ref{prop:first-reaction}, $0.7\%$ and $0.08\%$ under.

\mypara{The cost}  Measured on one machine, the social graph of Section~\ref{sec:example} at its mix: 100 agents for a year is 14,459 answers in 9.1 s of wall-clock time, 1,596 releases a second, at 80 MiB; 30 days at 100, 200 and 500 agents cost 29, 32 and 46 ms an agent and 2.3, 2.6 and 3.8 ms an answer; and 1,000 agents for five years is 740,622 answers in 2,291 s, 323 releases a second from start to end, at 555 MiB at the peak, the memory flat throughout.  The draw of a population of 100 costs about 2 s besides.

\mypara{How it was built}  The implementation --- the monitor, the person process, the two translations, the harnesses, the replays and the tests --- was written by Claude (Anthropic) from this paper, by the Human--Science--AI methodology~\cite{shapiro2026formalising,shapiro2026implementing}: the paper is the source of authority for coding, and every gap coding found was repaired in the paper before it was repaired in the code; the runs' analyses and the figures were made the same way.

\section{Example: The Grassroots Social Graph}\label{sec:example}

The platform's agent offers friendship to a peer its person picks, and puts each offer that arrives to its person as a card.  The agents are numbered $1$ to $n$; the first half are male and the rest female, and each agent is created with the list of the other agents of its sex and the list of those of the other sex, and starts by asking its person the first menu, by its exported procedure \verb|start|.

{\small\begin{verbatim}
Peer    ::= Constant.
YesNo   ::= yes ; no.
Peers   ::= [Peer | List(Peer)].
Choice  ::= same(Peer) ; other(Peer).
Menu    ::= menu(Peers, Peers, Choice?).
Card    ::= card(Peer, YesNo?).

%% Show the person both lists, and hand them the choice.
procedure (Menu)*ask(Peers?, Peers?, Choice).
(menu(Ss?, Os?, C))*ask(Ss, Os, C?).

%% Show the person who offers, and hand them yes or no.
procedure (Card)*respond(Peer?, YesNo).
(card(From?, A))*respond(From, A?).
\end{verbatim}}

The agent's own question is \verb|Menu|, in writer mode: \verb|ask| writes the two lists at once and hands the person the choice, which stays open while offers come and go, and each offer the person makes is followed by a new menu.  An arriving offer is \verb|Card|, also in writer mode: \verb|respond| shows the person who offers and hands them yes or no, and the reply carries their answer back when they give it.  The agent's clauses, \full{in the repository\footnote{\url{https://github.com/EShapiro2/GLP}, directory \texttt{programs/sglp}.}}{in the supplementary material}, send an offer to the peer the person picks and ask again, put each offer that arrives to the person and reply with their answer, and drop the replies to their own offers; both interactive types are in writer mode, so their person procedures take readers.

The profiles are two dimensions of two.  \verb|one_of(Xs?, S?, X)| draws \verb|X| uniformly from the list \verb|Xs| by \verb|random| from the seed \verb|S|, the list being non-empty by its type, as \verb|Peers| is; it is ordinary GLP, \full{in the repository}{in the supplementary material}.

{\small\begin{verbatim}
person homophile.
Menu =::= homophile_menu.
procedure homophile_menu(Menu?, Integer?).
homophile_menu(M, S) :-
    random(S?, 100, K, S1), pick(K?, S1?, M?) @ 1/week.
procedure pick(Integer?, Integer?, Menu?).
pick(K, S, menu(Ss, _, same(P?))) :-
    K? =< 98, ground(Ss?) | one_of(Ss?, S?, P).
pick(K, S, menu(_, Os, other(P?))) :-
    K? > 98, ground(Os?) | one_of(Os?, S?, P).

person wary.
Card =::= wary_card.
procedure wary_card(Card?, Integer?).
wary_card(C, S) :-
    random(S?, 100, K, _), answer(K?, C?) @ 1/day.
procedure answer(Integer?, Card?).
answer(K, card(_, yes)) :- K? =< 20 | true.
answer(K, card(_, no))  :- K? > 20  | true.

run 100 agents
    [ approach ~ (homophile : 0.6 ; indifferent : 0.4),
      response ~ (wary : 0.7 ; sociable : 0.3) ]
    until 1 year seed 20260927.
\end{verbatim}}

A homophile answers a menu after a week on average, from when it is asked, and picks a peer of their own sex 98 times in 100, uniformly from the list the menu shows; the indifferent is the homophile at \verb|2/week| and \verb|50|.  A wary person answers a card after a day on average and declines four times in five; the sociable is the wary at \verb|1/hour| and \verb|90|.  Approach and response make four profiles, and with sex eight kinds of agent.  Each agent starts by asking its person the first menu.

\begin{definition}[Friendship Graph of a Run]\label{def:friendship-graph}
The \temph{friendship graph} at simulated time $t$ of a run of this population is the graph on its agents $P$ whose edges are the $\{a, b\}$ for which the run's log (Definition~\ref{def:sglp-log}) has an entry $(t', a, \verb|card(|b\verb|, yes)|)$ or an entry $(t', b, \verb|card(|a\verb|, yes)|)$ with $t' < t$.
\end{definition}

An edge is a card answered \verb|yes| by either of its two agents before $t$, the card showing the other.

\begin{figure}[t]
\centering
\includegraphics[width=\columnwidth]{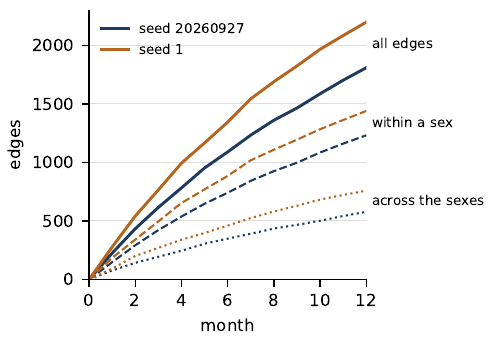}
\Description{A line chart of the edges of the friendship graph by month for two seeds, split into edges within a sex and across the sexes.}
\caption{The edges of the friendship graph at the end of each month, in two runs of 100 agents for one year: all edges, those within a sex and those across the sexes.}
\label{fig:social-graph-edges}
\end{figure}

\mypara{Results}  We ran this population with 100 agents for one year, with the seeds 20260927 and 1, each run nine seconds of wall-clock time (Section~\ref{sec:implementation}); 61 agents drew homophile and 74 wary with the first seed, 55 and 61 with the second.  With seed 20260927 the agents answered 7,235 menus, each answer an offer sent, and 7,224 cards, 2,767 of them yes; with seed 1, 7,523 menus and 7,507 cards, 3,543 of them yes.  Figure~\ref{fig:social-graph-edges} shows the edges of the friendship graph at the end of each month.  At the end of the year the graph has 1,810 edges with seed 20260927 and 2,199 with seed 1; it is one component from the third month on with the first seed and from the second with the other, and at the end of the first month two agents have no friend yet.  Of the 2,450 pairs within a sex, 50\% and 59\% are friends by the end of the year, and of the 2,500 pairs across the sexes, 23\% and 30\%: the indifferent answer a menu twice as often as the homophile, so 28\% and 32\% of the offers cross the sexes, and the edges across the sexes are 32\% and 35\% of the edges.  A yes to the card of one who is already a friend adds no edge; with seed 20260927 that was 730 of the 1,962 yeses within a sex and 227 of the 805 across the sexes.

\begin{table}[t]
\caption{The friendship graph at the end of the year by the mix of profiles, 100 agents, over the seeds 1 to 20: the edges and those across the sexes, mean $\pm$ standard deviation, and the month by whose end the graph is one component, mean and latest.}
\label{tab:social-graph-mix}
\footnotesize
\begin{tabular}{llrrr}
\toprule
homophile & wary & edges & across the sexes & one by month \\
\midrule
none & none & $4{,}199 \pm 26$ & $2{,}113 \pm 18$ & 1.0, 1 \\
none & 70\% & $2{,}634 \pm 159$ & $1{,}322 \pm 82$ & 1.3, 2 \\
none & all & $1{,}698 \pm 32$ & $845 \pm 25$ & 2.1, 3 \\
60\% & none & $3{,}213 \pm 88$ & $1{,}130 \pm 73$ & 1.0, 1 \\
60\% & 70\% & $1{,}926 \pm 147$ & $626 \pm 73$ & 2.2, 4 \\
60\% & all & $1{,}216 \pm 51$ & $385 \pm 37$ & 3.0, 5 \\
all & none & $2{,}163 \pm 22$ & $91 \pm 7$ & 1.1, 2 \\
all & 70\% & $1{,}335 \pm 74$ & $41 \pm 5$ & 2.4, 3 \\
all & all & $850 \pm 20$ & $20 \pm 4$ & 3.8, 7 \\
\bottomrule
\end{tabular}
\end{table}

\begin{figure}[t]
\centering
\includegraphics[width=\columnwidth]{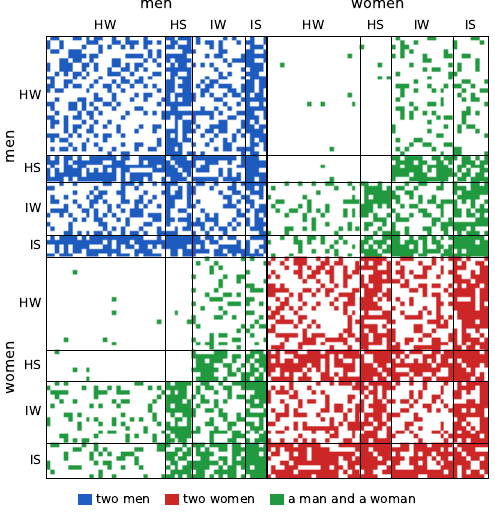}
\Description{An adjacency matrix of the friendship graph at the end of the year, agents ordered by sex and profile, edges coloured by the sexes of the pair.}
\caption{The friendship graph at the end of the year of the run with seed 20260927 as an adjacency matrix, a dot per edge: blue between two men, red between two women, green between a man and a woman.  The agents are ordered by sex and within a sex by profile: homophile wary (HW), homophile sociable (HS), indifferent wary (IW), indifferent sociable (IS).}
\label{fig:social-graph-matrix}
\end{figure}

\mypara{What the mix does}  Table~\ref{tab:social-graph-mix} has the friendship graph at the end of the year for each mix of none, 60\% or all homophile with none, 70\% or all wary, over twenty seeds, 1 to 20, the mean and the standard deviation of each count, and Figure~\ref{fig:social-graph-matrix} shows the graph of the run with seed 20260927 by profile: a sociable agent's row is dense, a wary one's sparse, and between two homophile of different sexes there is hardly an edge.  The edges go from about 850 with all homophile and all wary to about 4,200 with none of either, the spread over the seeds a few per cent of the mean at every mix: the homophile offer half as often as the indifferent and to half the peers, and the wary accept a fifth of the offers where the sociable accept nine in ten, though many of the sociable's yeses go to friends already, so all sociable in place of all wary has four and a half times the yeses and two and a half times the edges.  With all homophile the edges across the sexes are 2\% to 4\% of the edges and the graph within each sex is near complete, 2,072 edges of the 2,450 pairs with all sociable; with none homophile they are half the edges.  In every one of the 180 runs the graph is one component by the end of the year: by the end of the first month at every mix with none wary, and at the latest by the seventh, with all homophile and all wary.

\mypara{Scale}  The same population at 1,000 agents, 500 of each sex, for five simulated years, seed 20260927, is 740,622 answers in 38 minutes of wall-clock time (Section~\ref{sec:implementation}), 147,543 of the 370,245 cards answered yes.  The friendship graph has 28,211 edges at the end of the first year, 8,697 of them across the sexes, and 116,892 at the end of the fifth, 37,126 across; it is one component from the third month on, and at the end of the first month 19 agents have no friend yet.

\section{Example: Coins Among Friends}\label{sec:currencies}

The second platform is coins among friends, the coins-only restriction of grassroots currencies~\cite{shapiro2024gc,shapiro2026bonds}: a person's coins are units of their debt, a person holds only their friends' coins, and a swap at par between friends is mutual credit.  The agent's question is a menu showing its friends and its wallet and taking a request, to pay a friend or to propose a swap to one; a friend's proposal is a card showing who offers how many and how many of that friend's coins one holds, taking yes or no.  A payment sends the friend coins of their own, which they then no longer owe; an accepted swap issues to each party $K$ coins of the other's; and while a card waits for its answer the agent takes neither a message nor a request.  \verb|Peer|, \verb|Peers| and \verb|YesNo| are those of Section~\ref{sec:example}.

{\small\begin{verbatim}
Holding ::= holding(Peer, Integer).
Wallet  ::= List(Holding).
Request ::= pay(Peer, Integer) ; swap(Peer, Integer).
Menu    ::= menu(Peers, Wallet, Request?).
Offer   ::= offer(Peer, Integer, Integer, YesNo?).

%% Show the person their friends and wallet, and hand them
%% the request.
procedure (Menu)*ask(Peers?, Wallet?, Request).
(menu(Fs?, W?, R))*ask(Fs, W, R?).

%% Show the person who offers how many, and how many of
%% theirs they hold; hand them yes or no.
procedure (Offer)*respond(Peer?, Integer?, Integer?, YesNo).
(offer(From?, K?, Held?, A))*respond(From, K, Held, A?).
\end{verbatim}}

The agent's clauses, \full{in the repository}{in the supplementary material}, send a friend the person pays what the wallet holds of that friend's coins, up to the amount asked, and ask again; send a friend the person chooses a proposal of $K$ and ask again; put a proposal that arrives to the person, and on yes hold $K$ of the proposer's coins and send them $K$ of the agent's own, on no send a decline; hold coins that arrive from a friend, and hold nothing for a payment, which returns the agent's own coins; and drop the declines of their own proposals.

The profiles are two dimensions of two; \verb|one_of| is that of Section~\ref{sec:example}:

{\small\begin{verbatim}
%% Pay a friend whose coins I hold one to five; holding none,
%% propose a swap of ten to a friend.
procedure request(Integer?, Menu?).
request(S, menu(Fs, [], swap(F?, 10))) :-
    ground(Fs?) | one_of(Fs?, S?, F).
request(S, menu(_, [H|Hs], pay(F?, K?))) :-
    ground(H?), ground(Hs?) |
    random(S?, 5, K, S1), one_of([H?|Hs?], S1?, Hd),
    issuer(Hd?, F).

person spender.
Menu =::= spender_menu.
procedure spender_menu(Menu?, Integer?).
spender_menu(M, S) :- request(S?, M?) @ 1/day.

person cautious.
Offer =::= cautious_offer.
procedure cautious_offer(Offer?, Integer?).
cautious_offer(O, _) :- weigh(O?) @ 1/hour.
procedure weigh(Offer?).
weigh(offer(_, _, Held, yes)) :- Held? =< 20 | true.
weigh(offer(_, _, Held, no))  :- Held? > 20  | true.

run 100 agents
    [ spending ~ (spender : 0.5 ; saver : 0.5),
      credit ~ (cautious : 0.5 ; generous : 0.5) ]
    until 1 year seed 20260927.
\end{verbatim}}

A spender acts daily and a saver weekly, with the same \verb|request|: each pays a friend whose coins they hold one to five, and holding none proposes a swap of ten to a friend drawn from the list; \verb|issuer| reads the issuer of a holding.  The cautious accept a proposal while holding at most twenty of the proposer's coins, and the generous accept one as it comes, each within an hour on average.  Each agent is created with its identifier, its friends in the friendship graph at one year of the run of Section~\ref{sec:example}, and an empty wallet, and starts by asking its person the first menu, by its exported procedure \verb|start|.

\begin{definition}[Wallet, Circulation]\label{def:circulation}
Let $W_a(c)$ be the term at the argument of type \verb|Wallet| of the goal of the procedure \verb|agent| at agent $a$ in a configuration $c$, the \temph{wallet} of $a$ in $c$.  The \temph{circulation} at simulated time $t$ of a run $r$ of this population is $\sum_{a \in P} \sum_{\texttt{holding(}\_\texttt{, }k\texttt{)} \in W_a(c)} k$, $c$ the last configuration of $r$ with simulated time at most $t$.
\end{definition}

The configuration $c$ is quiescent, the clock advancing only at a Release (Definition~\ref{def:sglp-ts}), so every wallet in it is ground.  Every holding is of another's coins, so the circulation is the debt the agents hold of one another.

\mypara{Results}  We ran this population with 100 agents for one year over the friendship graph at one year of Section~\ref{sec:example}'s run of the same seed, with the seeds 20260927 and 1 (1,810 and 2,199 edges), each run about ten seconds of wall-clock time.  With seed 20260927 the agents paid 17,232 times and proposed 2,908 swaps, 2,840 accepted and 68 declined; with seed 1, 15,873 pays and 2,726 proposals, 2,693 accepted and 33 declined.  An accepted swap of ten issues twenty coins and a payment returns what it pays to its issuer, so the circulation is what has been issued less what has been paid: it is 2,567 and 2,837 coins at the end of the first month and grows by about 700 and 800 a month to 10,467 and 11,355 at the end of the year, held in 711 and 830 holdings by 97 of the 100 agents with the first seed and by all of them with the second.  The cautious decline a proposal from a friend of whose coins they hold more than twenty, so declines come as holdings grow: 68 of the 2,908 proposals were declined with seed 20260927, 57 of them in the second half of the year.  Figure~\ref{fig:coins-holdings} shows who holds whose coins in the run with seed 20260927, in which 48 agents drew spender and 59 cautious: at the end of the year 9,754 of the 10,467 coins are spenders' coins in savers' wallets, the spenders having issued 10,109 of the coins out and the savers holding 9,916 of them, and the two run together from the first month; the other three pairs hold about seven hundred between them.  Figure~\ref{fig:coins-matrix} shows the holdings one by one at the end of that year: the generous savers hold 6,341 coins and the cautious savers 3,575, where no cautious holding passes thirty, since the cautious decline above twenty; the cautious spenders have 7,164 coins out and the generous spenders 2,945; and the two spender profiles hold about five hundred and fifty between them.

\begin{figure}[t]
\centering
\includegraphics[width=\columnwidth]{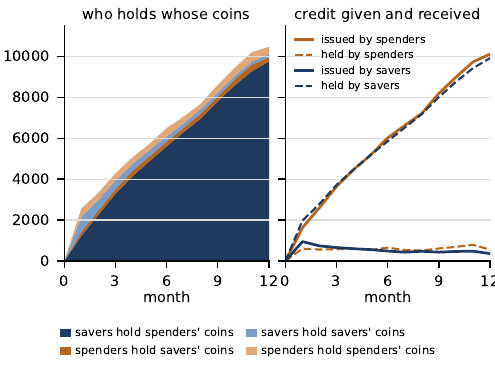}
\Description{Two charts by month: the circulation split by which profile holds which profile's coins, and the coins issued and held by spenders and by savers.}
\caption{The run with seed 20260927, 100 agents, half spenders and half cautious, by month.  Left: the circulation by who holds whose coins.  Right: the coins issued by the spenders and by the savers, and those held by each --- the credit each profile has given and received.}
\label{fig:coins-holdings}
\end{figure}

\begin{figure}[t]
\centering
\includegraphics[width=\columnwidth]{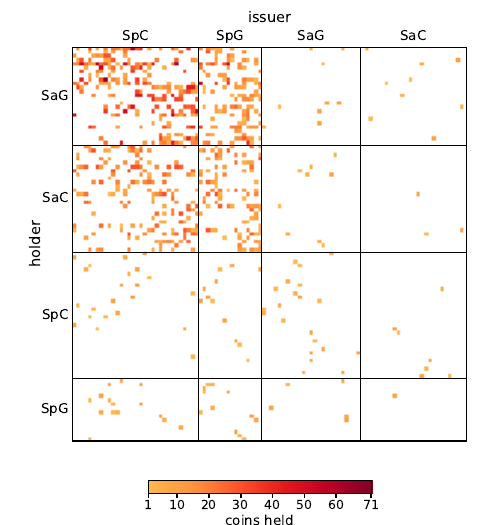}
\Description{A matrix of who holds whose coins at the end of the year, holders by row and issuers by column, ordered by profile, coloured by the coins held.}
\caption{Who holds whose coins at the end of the year of the run with seed 20260927: holders by row, issuers by column, a cell the coins held, white where none.  The agents are ordered by profile --- spender (Sp) or saver (Sa), cautious (C) or generous (G) --- the holders' profiles by the coins they hold and the issuers' by the coins they have out, most first.}
\label{fig:coins-matrix}
\end{figure}

\begin{table}[t]
\caption{Coins among friends at the end of the year by the mix of profiles, 100 agents over the friendship graphs of Section~\ref{sec:example}, over the seeds 1 to 20, mean $\pm$ standard deviation: the circulation, the proposals made and those declined.}
\label{tab:coins-mix}
\footnotesize
\begin{tabular}{llrrr}
\toprule
spenders & cautious & circulation & proposals & declined \\
\midrule
none & none & $1{,}660 \pm 116$ & $687 \pm 10$ & $0$ \\
none & half & $1{,}660 \pm 116$ & $687 \pm 10$ & $0$ \\
none & all & $1{,}655 \pm 115$ & $687 \pm 10$ & $0.1 \pm 0.2$ \\
half & none & $11{,}546 \pm 722$ & $2{,}917 \pm 131$ & $0$ \\
half & half & $11{,}282 \pm 797$ & $2{,}961 \pm 143$ & $64 \pm 25$ \\
half & all & $10{,}328 \pm 621$ & $3{,}030 \pm 150$ & $190 \pm 49$ \\
all & none & $1{,}702 \pm 148$ & $4{,}369 \pm 27$ & $0$ \\
all & half & $1{,}683 \pm 155$ & $4{,}368 \pm 27$ & $0.1 \pm 0.3$ \\
all & all & $1{,}675 \pm 135$ & $4{,}370 \pm 29$ & $1.2 \pm 1.7$ \\
\bottomrule
\end{tabular}
\end{table}

\begin{figure}[t]
\centering
\includegraphics[width=\columnwidth]{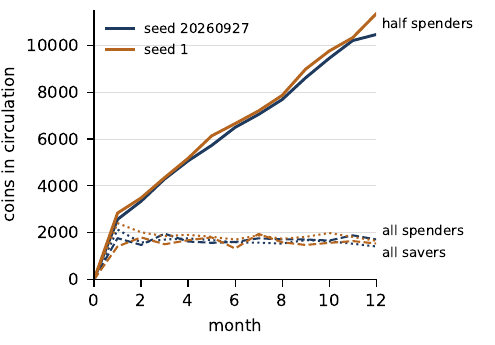}
\Description{A line chart of the coins in circulation by month for none, half and all spenders, two seeds.}
\caption{The circulation at the end of each month with half the agents cautious and none, half or all of them spenders, 100 agents for one year over the friendship graphs of Section~\ref{sec:example}, seeds 20260927 and 1.}
\label{fig:coins-circulation}
\end{figure}

\mypara{What the mix does}  Table~\ref{tab:coins-mix} has the circulation at the end of the year for each mix of none, half or all spenders with none, half or all cautious, over the seeds 1 to 20, each over the friendship graph of Section~\ref{sec:example}'s run of the same seed, and Figure~\ref{fig:coins-circulation} the circulation by month with half cautious, for the two seeds of the figures.  The circulation is highest where spenders and savers mix: about 11,300 coins at half and half, within 800 over the seeds, and still growing, against about 1,700 with all savers or all spenders.  A payment carries a friend's own coins back to them, and a spender pays daily, so with all spenders a swap's twenty coins are paid back within days and after 4,369 swaps 1,702 coins are out; a saver pays weekly, so with all savers few swaps are made and few coins are out; and at half and half a saver holds for weeks the coins of the spenders who pay it, which is where the circulation is.  The credit dimension hardly bites at these holdings: with all savers or all spenders a holding passes twenty in a run or two of the twenty, so those runs are the same whatever the credit mix but for a decline or two, and only with half the agents spenders do the cautious decline, 64 of 2,961 proposals on average, and 190 of 3,030 with all cautious.

\section{Related Work}\label{sec:related-work}

Six lines of work meet here: process calculi and languages with rates, simulated exactly; rule-based languages whose simulators do not enumerate the species they generate; languages and formalisms for continuous-time agent-based simulation; the agent-based platforms and the testing of multiagent systems in simulation; simulated people with designated profiles; and what the two examples study, homophily in the formation of friendships and the circulation of mutual credit.  sGLP differs from all of them in one way: it simulates the program the platform deploys, and its stochastic specification covers the person.  Appendix~\ref{app:related} expands each line.

\mypara{Calculi and languages with rates} The stochastic pi-calculus attaches a rate to each prefix~\cite{priami1995stochastic}; its biochemical variant moves the rate onto the channel and settles the race among the enabled communications by Gillespie's direct method~\cite{priami2001application,gillespie1977exact}, and BioSPI implemented it in Flat Concurrent Prolog~\cite{shapiro1989family,regev2000representation,priami2001application,silverman1988logix}, a monitor drawing the channel and advancing the clock, as ours does.  The stochastic pi-machine groups identical processes and counts the copies, and is proved reduction-equivalent to the calculus~\cite{phillips2007efficient}; a later machine carries the grouping to restricted channels~\cite{pauleve2010generic}.  Bio-PEPA is a stochastic process algebra with the same simulation~\cite{ciocchetta2009biopepa}, and stochastic concurrent constraint programming attaches rates to a concurrent constraint language and simulates it by Gillespie's method~\cite{bortolussi2008stochastic}, the nearest language to sGLP.  The grouping is not available to us, since no two agents of a grassroots platform share a state; the proof we meet in the same form, Theorem~\ref{thm:soundness}.

\mypara{Rule-based languages} Kappa's simulator does not count species at all~\cite{danos2007scalable}: it over-approximates the next events, rejects the draws that apply no rule, corrects the time advance, and wakes only the rules a static activation relation names; in sGLP the rules are the interactive types, few, and the instances the open questions, many, and the activation relation follows from the semantics, since a reduction reaches one agent.

\mypara{Continuous-time agent-based simulation} ML3 specifies an agent by stochastic rules of a guard, a rate and an effect, alternative decisions competing by stochastic race, and simulates by the direct method with a dynamic dependency graph~\cite{koster2024fast}; DEVS is the discrete-event formalism such simulators rest on~\cite{zeigler2000theory}.  ML3 is the nearest prior work to sGLP, and the difference is what is simulated: there the rules are the model, here a vGLP program deployed as it stands, and the stochastic specification covers the person's answers.

\mypara{Agent-based platforms, and testing in simulation} NetLogo, Repast, MASON, Mesa and krABMaga schedule every agent at each step of a discrete clock~\cite{antelmi2024reliable}, where a stochastic race has no formulation~\cite{koster2024fast}, and none has a place for the person.  Testing a multiagent system by simulating the code it deploys, with no surrogate, is argued for BDI agents mapped onto a discrete-event simulator~\cite{baiardi2026testing}; there the agents are the system and no person is simulated, where here the program is the platform's and the person is what is drawn.

\mypara{Simulated people with profiles} Generative agents give each simulated person a text persona and have a language model answer for them~\cite{park2023generative}, OASIS runs a million such agents on a social platform~\cite{yang2024oasis}, and user simulation evaluates interactive systems by simulated users with personas~\cite{balog2024user,liu2026personaeval}.  An sGLP profile is a program with rates, typed by the dual of the question it answers, drawn from a seed and reproducible, and the simulation is a run of the deployed program; a persona is a text and its answers are a model's.

\mypara{What the examples study} Homophily, the tendency to befriend one's like, is the best-attested regularity of social networks~\cite{mcpherson2001birds}, and Currarini, Jackson and Pin model friendship formation under it, with the mixing set by a group's size and bias~\cite{currarini2009economic}; Section~\ref{sec:example}'s runs are such a model's dynamics, with the bias a profile and the platform the one deployed.  Mutual credit is studied on the Sardex network, its transactions' cycles and the liquidity they save~\cite{iosifidis2018sardex,fleischman2020liquidity,littera2017sardex}, and modelled for LETS~\cite{schraven2001mutual}; Section~\ref{sec:currencies} simulates coins among friends over the friendship graph the social graph formed, with spending and credit as profiles.

\mypara{The exponential assumption} The inter-event times of human activity are heavier-tailed than exponential; the non-Markovian and the Laplace Gillespie algorithms simulate renewal processes with other distributions~\cite{masuda2022gillespie}, and nothing in sGLP fixes the exponential: drawing a rated goal's activation time from another distribution changes the draw and leaves the language as it is.

\mypara{What is new} To the best of our knowledge, sGLP is the first language in which the person of a deployed multiagent program is specified as a stochastic agent answering the questions that program puts to them, so that one program is simulated at scale and deployed on a phone without being rewritten, and the first in which the stochastic specification is confined to the person while the machine keeps a nondeterministic semantics with no scheduling assumption.

\section{Conclusion}\label{sec:conclusion}

sGLP is GLP with a rate on a goal, and a stochastic agent is an sGLP program in place of the UI of a vGLP platform and the person behind it, holding the person's end of each question the platform asks, typed by the dual of the question's interactive type.  The platform's program is unchanged; a population is a mix of profiles of person, each a program; and every simulated run is a run of the platform in which every answer was given by a stochastic agent, so what a simulation shows is something the platform can do.

Here we employed an exponential distribution; another distribution can be employed without changing the language~\cite{masuda2022gillespie}. 

The simulations done here show that platform behaviour can be tested and studied by mixing the stochastic profiles of its people: the friendship graph that forms, and how many coins circulate.

Four limitations.  The simulator hosts the population in one machine, where the guard \verb|when_idle| is sound (Section~\ref{sec:implementation}).  A rate per act is a modelling assumption, and the inter-event times of people are heavier-tailed than exponential (Section~\ref{sec:related-work}).  The profiles and their rates are written, not estimated from how people use the deployed platforms.  Reading the logs and analysing the runs is software outside the language (Section~\ref{sec:simulation}).

\bibliography{bib}

@misc{shapiro2026bonds,
  author = {Ehud Shapiro},
  title  = {Grassroots Bonds: Financing by the People, for the People},
  year   = {2026},
  howpublished = {arXiv:2603.13671},
  eprint = {2603.13671},
  archivePrefix = {arXiv},
  url    = {https://arxiv.org/abs/2603.13671}
}

@misc{shapiro2026volition,
      title={Volition Elicitation:  Operational Semantics for People and Their Machines}, 
      author={Ehud Shapiro},
      year={2026},
      howpublished={arXiv:2607.14138},
      eprint={2607.14138},
      archivePrefix={arXiv},
      primaryClass={cs.DC},
      url={https://arxiv.org/abs/2607.14138}, 
}

@article{littera2017sardex,
  author  = {Littera, Giuseppe and Sartori, Laura and Dini, Paolo and Antoniadis, Panayotis},
  title   = {From an Idea to a Scalable Working Model: Merging Economic Benefits with Social Values in {Sardex}},
  journal = {International Journal of Community Currency Research},
  volume  = {21},
  pages   = {6--21},
  year    = {2017},
  doi     = {10.15133/j.ijccr.2017.002}
}

@article{lewis2026volitional,
  title={Volition-Guarded Multiagent Atomic Transactions: Describing People and their Machines},
  author={Lewis-Pye, Andy and Shapiro, Ehud},
  journal={arXiv preprint arXiv:2604.25596},
  year={2026},
  url = {https://arxiv.org/abs/2604.25596}
}

@article{shapiro2025characterising,
  title={Characterising Global Platforms: Centralised, Decentralised, Federated, and Grassroots},
  author={Shapiro, Ehud},
  journal={arXiv preprint arXiv:2511.03286},
  year={2025},
  url = {https://arxiv.org/abs/2511.03286}
}

@article{shapiro2026types,
  title={Types for Grassroots Logic Programs (Full version)},
  author={Shapiro, Ehud},
  journal={arXiv preprint arXiv:2601.17957},
  year={2026},
  url = {https://arxiv.org/abs/2601.17957}
}

@article{shapiro2026implementing,
  title={Implementing Grassroots Logic Programs with Multiagent Transition Systems and AI (Full Version)},
  author={Shapiro, Ehud},
  journal={arXiv:2602.06934.  Summary to appear in Proc. of LOPSTR+PPDP'26},
  year={2026},
  url = {https://arxiv.org/abs/2602.06934}
}

@misc{dart2024,
     author       = {{Google}},
     title        = {Dart Programming Language},
     year         = {2024},
  url = {https://dart.dev/}
}

@article{shapiro2025glp,
  title={GLP: A Grassroots, Multiagent, Concurrent, Logic Programming Language for AI (full version)},
  author={Shapiro, Ehud},
  journal={arXiv preprint arXiv:2510.15747, Summary in Proc. of ICLP'26, EPTCS 450, pp. 119--133},
  year={2025},
  url = {https://arxiv.org/abs/2510.15747}
}

@incollection{regev2000representation,
  title={Representation and simulation of biochemical processes using the $\pi$-calculus process algebra},
  author={Regev, Aviv and Silverman, William and Shapiro, Ehud},
  booktitle={Biocomputing 2001},
  pages={459--470},
  year={2000},
  publisher={World Scientific},
  doi = {10.1142/9789814447362_0045}
}

@inproceedings{shapiro2023gsn, 
author={Shapiro, Ehud},
title={Grassroots Social Networking: Serverless, Permissionless Protocols for Twitter/LinkedIn/WhatsApp},
year = {2023},
isbn = {979-8-4007-0225-9/23/09},
publisher = {Association for Computing Machinery},
doi = {10.1145/3599696.3612898},
booktitle = {OASIS ’23},
location = {Rome, Italy},
}

@inproceedings{shapiro2023grassrootsBA, 
author={Shapiro, Ehud},
 title={Grassroots Distributed Systems: Concept, Examples, Implementation and Applications (Brief Announcement)
},year = {2023},
publisher = {LIPICS},
booktitle = {37th International Symposium on Distributed Computing (DISC 2023). (Extended version: arXiv:2301.04391)},
address = {Italy},
notes={Extended version: arXiv preprint arXiv:2301.04391},
pages = {47:1--47:7},
  url = {https://arxiv.org/abs/2301.04391}
}

@article{shapiro2024grassroots,
  title={A Grassroots Architecture to Supplant Global Digital Platforms by a Global Digital Democracy},
  author={Shapiro, Ehud},
  journal={arXiv:2404.13468, Proceedings of DAWO'24},
  year={2024},
  url = {https://arxiv.org/abs/2404.13468}
}

@article{shapiro2024gc,
  title={Grassroots Currencies: Turning Mutual Trust into Liquidity},
  author={Shapiro, Ehud},
  journal={arXiv preprint arXiv:2202.05619},
  year={2022},
  url = {https://arxiv.org/abs/2202.05619}
}

@article{shapiro1989family,
  doi = {10.1145/72551.72555},
  title={The family of concurrent logic programming languages},
  author={Shapiro, Ehud},
  journal={ACM Computing Surveys (CSUR)},
  volume={21},
  number={3},
  pages={413--510},
  year={1989},
  publisher={ACM New York, NY, USA}
}

@article{shapiro2021multiagent,
  title={Multiagent Transition Systems for Composing Fault-Resilient Protocol Stacks},
  author={Shapiro, Ehud},
  journal={arXiv preprint arXiv:2112.13650},
  year={2021},
  url = {https://arxiv.org/abs/2112.13650}
}

@misc{proof,
  title={Proof of Stake FAQ, Ethereum Wiki},
  author={Nakamoto, Satoshi and Bitcoin, A},
  howpublished = { \\ \MYhref{https://eth.wiki/en/concepts/proof-of-stake-faqs}{https://eth.wiki/en/concepts/proof-of-stake-faqs}},
  year={2019}
}

@incollection{silverman1988logix,
  title={The Logix system user manual Version 1.21},
  author={Silverman, William and Hirsch, Michael and Houri, Avshalom and Shapiro, Ehud},
  booktitle={Concurrent Prolog: Collected Papers},
  pages={46--77},
  year={1988},
  publisher={MIT Press}
}

@article{priami2001application,
  title={Application of a stochastic name-passing calculus to representation and simulation of molecular processes},
  author={Priami, Corrado and Regev, Aviv and Shapiro, Ehud and Silverman, William},
  journal={Information processing letters},
  volume={80},
  number={1},
  pages={25--31},
  year={2001},
  publisher={Elsevier},
  doi = {10.1016/s0020-0190(01)00214-9}
}

@article{iosifidis2018sardex,
  author  = {Iosifidis, George and Charette, Yanick and Airoldi, Edoardo and Littera, Giuseppe and Tassiulas, Leandros and Christakis, Nicholas},
  title   = {Cyclic Motifs in the {Sardex} Monetary Network},
  journal = {Nature Human Behaviour},
  volume  = {2},
  pages   = {822--829},
  year    = {2018},
  doi     = {10.1038/s41562-018-0450-0}
}

@misc{shapiro2026gsg,
  author = {Ohad Eitan and Idit Keidar and Ehud Shapiro},
  title  = {The Grassroots Social Graph: A Secure Operating System for Multiagent, Distributed, Peer-to-Peer, Smartphone-Based Platforms},
  year   = {2026},
  note   = {In preparation}
}

@article{shapiro2026formalising,
  author  = {James Golike and Andy Lewis-Pye and Ehud Shapiro},
  title   = {Formalising and Implementing Grassroots Social Contracts: From Legal Text to Grassroots Platforms (Full Version)},
  year    = {2026},
  journal = {arXiv:2609.06231}
}

@inproceedings{park2023generative,
  author    = {Joon Sung Park and Joseph C. O'Brien and Carrie Jun Cai and Meredith Ringel Morris and Percy Liang and Michael S. Bernstein},
  title     = {Generative Agents: Interactive Simulacra of Human Behavior},
  booktitle = {Proceedings of the 36th Annual ACM Symposium on User Interface Software and Technology (UIST 2023)},
  pages     = {1--22},
  year      = {2023},
  doi       = {10.1145/3586183.3606763}
}

@article{priami1995stochastic,
  title={Stochastic $\pi$-calculus},
  author={Priami, Corrado},
  journal={The Computer Journal},
  volume={38},
  number={7},
  pages={578--589},
  year={1995},
  publisher={Oxford University Press}
}

@article{gillespie1977exact,
  title={Exact stochastic simulation of coupled chemical reactions},
  author={Gillespie, Daniel T.},
  journal={The Journal of Physical Chemistry},
  volume={81},
  number={25},
  pages={2340--2361},
  year={1977},
  publisher={American Chemical Society},
  doi={10.1021/j100540a008}
}

@inproceedings{pauleve2010generic,
  title={A Generic Abstract Machine for Stochastic Process Calculi},
  author={Paulev{\'e}, Lo{\"i}c and Youssef, Simon and Lakin, Matthew R. and Phillips, Andrew},
  booktitle={Proceedings of the 8th International Conference on Computational Methods in Systems Biology (CMSB '10)},
  pages={43--54},
  year={2010},
  publisher={ACM},
  doi={10.1145/1839764.1839771}
}

@inproceedings{danos2007scalable,
  title={Scalable Simulation of Cellular Signaling Networks},
  author={Danos, Vincent and Feret, J{\'e}r{\^o}me and Fontana, Walter and Krivine, Jean},
  booktitle={Proceedings of the 5th Asian Symposium on Programming Languages and Systems (APLAS'07)},
  pages={139--157},
  year={2007},
  publisher={Springer},
  doi={10.1007/978-3-540-76637-7_10}
}

@misc{masuda2022gillespie,
  title={Gillespie Algorithms for Stochastic Multiagent Dynamics in Populations and Networks},
  author={Masuda, Naoki and Vestergaard, Christian L.},
  note={Elements in the Structure and Dynamics of Complex Networks, Cambridge University Press},
  year={2021},
  eprint={2112.05293},
  archivePrefix={arXiv}
}

@article{koster2024fast,
  title={A Fast Embedded Language for Continuous-Time Agent-Based Simulation},
  author={K{\"o}ster, Till and Reinhardt, Oliver and Hinsch, Martin and Bijak, Jakub and Uhrmacher, Adelinde M.},
  journal={Journal of Artificial Societies and Social Simulation},
  volume={27},
  number={1},
  pages={10},
  year={2024},
  doi={10.18564/jasss.5232}
}

@article{antelmi2024reliable,
  title={Reliable and Efficient Agent-Based Modeling and Simulation},
  author={Antelmi, Alessia and Caramante, Pasquale and Cordasco, Gennaro and D'Ambrosio, Giuseppe and De Vinco, Daniele and Foglia, Francesco and Postiglione, Luca and Spagnuolo, Carmine},
  journal={Journal of Artificial Societies and Social Simulation},
  volume={27},
  number={2},
  pages={4},
  year={2024},
  doi={10.18564/jasss.5300}
}

@inproceedings{phillips2007efficient,
  title={Efficient, Correct Simulation of Biological Processes in the Stochastic Pi-calculus},
  author={Phillips, Andrew and Cardelli, Luca},
  booktitle={Computational Methods in Systems Biology (CMSB 2007)},
  series={Lecture Notes in Bioinformatics},
  volume={4695},
  pages={184--199},
  year={2007},
  publisher={Springer},
  doi={10.1007/978-3-540-75140-3_13}
}

@article{schraven2001mutual,
  title={Mutual Credit Systems and the Commons Problem: Why Community Currency Systems such as {LETS} Need Not Collapse Under Opportunistic Behaviour},
  author={Schraven, Jorim},
  journal={International Journal of Community Currency Research},
  volume={5},
  number={4},
  year={2001},
  doi={10.15133/j.ijccr.2001.002}
}

@article{fleischman2020liquidity,
  title={Liquidity-Saving through Obligation-Clearing and Mutual Credit: An Effective Monetary Innovation for {SMEs} in Times of Crisis},
  author={Fleischman, Toma{\v{z}} and Dini, Paolo and Littera, Giuseppe},
  journal={Journal of Risk and Financial Management},
  volume={13},
  number={12},
  pages={295},
  year={2020},
  doi={10.3390/jrfm13120295}
}

@article{baiardi2026testing,
  title={Testing {BDI}-based Multi-Agent Systems using Discrete Event Simulation},
  author={Baiardi, Martina and Burattini, Samuele and Ciatto, Giovanni and Pianini, Danilo},
  journal={arXiv preprint arXiv:2602.13878},
  year={2026},
  url={https://arxiv.org/abs/2602.13878}
}

@article{bortolussi2008stochastic,
  title={Modeling Biological Systems in Stochastic Concurrent Constraint Programming},
  author={Bortolussi, Luca and Policriti, Alberto},
  journal={Constraints},
  volume={13},
  number={1--2},
  pages={66--90},
  year={2008},
  doi={10.1007/s10601-007-9034-8}
}

@article{ciocchetta2009biopepa,
  title={{Bio-PEPA}: a Framework for the Modelling and Analysis of Biochemical Networks},
  author={Ciocchetta, Federica and Hillston, Jane},
  journal={Theoretical Computer Science},
  volume={410},
  number={33--34},
  pages={3065--3084},
  year={2009}
}

@book{zeigler2000theory,
  title={Theory of Modeling and Simulation: Integrating Discrete Event and Continuous Complex Dynamic Systems},
  author={Zeigler, Bernard P. and Praehofer, Herbert and Kim, Tag Gon},
  edition={2nd},
  publisher={Academic Press},
  year={2000}
}

@article{yang2024oasis,
  title={{OASIS}: Open Agent Social Interaction Simulations with One Million Agents},
  author={Yang, Ziyi and Zhang, Zaibin and Zheng, Zirui and Jiang, Yuxian and Gan, Ziyue and Wang, Zhiyu and Ling, Zijian and Chen, Jinsong and Ma, Martz and Dong, Bowen and Gupta, Prateek and Hu, Shuyue and Yin, Zhenfei and Li, Guohao and Jia, Xu and Wang, Lijun and Ghanem, Bernard and Lu, Huchuan and Lu, Chaochao and Ouyang, Wanli and Qiao, Yu and Torr, Philip and Shao, Jing},
  journal={arXiv preprint arXiv:2411.11581},
  year={2024},
  url={https://arxiv.org/abs/2411.11581}
}

@article{balog2024user,
  title={User Simulation for Evaluating Information Access Systems},
  author={Balog, Krisztian and Zhai, ChengXiang},
  journal={Foundations and Trends in Information Retrieval},
  volume={18},
  number={1--2},
  pages={1--261},
  year={2024},
  doi={10.1561/1500000098}
}

@article{liu2026personaeval,
  title={{PersonaEval}: Persona-Based User Simulation for Evaluating Interactive Applications},
  author={Liu, Yifan Simon and Wen, Qianfeng and Fan, Yilan and Huang, Shirley and Gao, Ruoqi and Hou, Jianheng and Mohsin, Muhammad Ahmed and Di, Zonglin and Joshi, Brihi and Tan, Xincheng and Lu, Yucheng and Liu, Xiaoyi and Liu, Heming and Xing, Hanwen and Min, Guanghui and Shan, Zhengyang and Nguyen, My Chiffon and Gupta, Ishan and Xiao, Yunze and Collison, Hannah and Huang, Jintao and Li, Jiatong and Jajee, Sankalp and Zhao, Yunhan and Hu, Bing and Ng, Sky and Chen, Xupeng and Lu, Binghang and Xiao, Weihang and Mohan, Aravind and Sun, Bolun and Wu, Yunshu and Xu, Yuanda and Shen, Yun and Zhang, Runyu and Deng, Zheyuan and Zhang, Zhiwei and Zhu, Qianyu and Wang, Dianzhuo and Wang, Yijun and He, Yixuan and Hao, Yuexing and Li, Xiaomin},
  journal={arXiv preprint arXiv:2608.15838},
  year={2026},
  url={https://arxiv.org/abs/2608.15838}
}

@article{mcpherson2001birds,
  title={Birds of a Feather: Homophily in Social Networks},
  author={McPherson, Miller and Smith-Lovin, Lynn and Cook, James M.},
  journal={Annual Review of Sociology},
  volume={27},
  pages={415--444},
  year={2001},
  doi={10.1146/annurev.soc.27.1.415}
}

@article{currarini2009economic,
  title={An Economic Model of Friendship: Homophily, Minorities, and Segregation},
  author={Currarini, Sergio and Jackson, Matthew O. and Pin, Paolo},
  journal={Econometrica},
  volume={77},
  number={4},
  pages={1003--1045},
  year={2009},
  doi={10.3982/ECTA7528}
}

\clearpage
\appendix
\section{Proofs}\label{app:proofs}

\soundnesssetting
\begin{theoremrestated}
\soundnessstatement
\end{theoremrestated}
\begin{proof}
Let $r = (c^1,t_1,\alpha_1), (c^2,t_2,\alpha_2), \ldots$ with $(c^1,t_1,\alpha_1) = (c^0,0,\emptyset)$.  Throughout, a variable of $I_a$ occurs nowhere before the transition that introduces it, by the properness of the maGLP computation $c^1, c^2, \ldots$ (Definition~\ref{def:sglp-run}); every clause of $N$ is SRSW (Definition~\ref{def:rc-glp-program}); and a person goal holds a variable of $I_a$ only as the person process holding an ask or as a relay, since the person process passes the term of each ask to the relay it spawns and a relay puts no variable it receives into the copy (Definition~\ref{def:relay}), so that the goals descended from the person procedures' goals hold the copies' variables and none of $I_a$.

\emph{Safety.}  We show that each transition $(c^i,t_i,\alpha_i) \rightarrow (c^{i+1},t_{i+1},\alpha_{i+1})$ of $r$ maps under $\sigma$ to a transaction of the vmaGLP transition system over $P$ and $M$ or to a stutter; then $\sigma(r)$ is a safe computation from $\sigma(c^0)$ (Definitions~\ref{def:rc-ts} and~\ref{def:rc-impl}).  A Release has $c^{i+1} = c^i$: a stutter.  A Machine transition is a Reduce, a Communicate or a Cold-call of the maGLP transition system over $P$ and $N$ (Definition~\ref{def:sglp-ts}), and we take each by what it reduces or delivers.

(1) A Reduce at $a$ of a platform goal with an asking clause of $\lceil M\rceil$ replaces $q(S_1,\ldots,S_n,D)$ by $q_1(S'_1,\ldots,S'_n,X?,D')$ or $q_1(S'_1,\ldots,S'_n,X,D')$, $X$ fresh, and adds to $\sigma_a$ the assignment of $D$ to the list of the ask \verb|ask(|$T$\verb|, |$t$\verb|(|$E$\verb|))| and the tail $D'?$, $E$ the other end of $X$, and the assignments by which each $S'_l$ passes its argument on (Definition~\ref{def:rc-compilation}).  In the image, $G$ loses $q(S_1,\ldots,S_n)$ and gains the asked goal, read through the pairs of its arguments and with the identifier of the goal reduced; $H$ gains $E$, a variable of $I_a$ that now occurs in an ask in $\sigma_a$; $\sigma_r$ is unchanged, $D$ being an ask-stream variable, the assignments of the pairs being outside it, and the assignment $W? := S_l?$ of an output position, read through, being absent while $S_l?$ is unassigned (Definition~\ref{def:simulation-sigma}).  This is the Ask of $a$'s goal with interactive variable $X$.

(2) A Reduce at $a$ of a platform goal with a compiled clause of $M$: $G$ changes as the Reduce of the source goal with the source clause does (Lemma~\ref{lem:rc-correspondence}), the platform goals and clauses of $N$ being those of $\lceil M\rceil$ and read back as there; $\sigma_r$ gains the assignments the source reduction adds, read through the pairs as there, an assignment to a variable of $I_a$ at a position the program writes, whose reader a relay holds, among them where the clause writes into a question; and $H$ is unchanged, since the reduction touches no person goal and no ask, and a variable it introduces occurs in neither.  This is the source Reduce.  A Reduce of a platform goal with a \verb|merge| clause the compilation adds, which moves asks between ask streams, changes nothing $\sigma$ reads: a stutter.

(3) A Reduce at $a$ of a person goal that adds no answer to a variable of $I_a$ --- the person process, a goal descended from it other than a relay, or a relay that assigns the copy, gives a writer's place to another, or assigns nothing received: $G$ and $\sigma_r$ are unchanged, a writer that takes another's place being read as its question; and so is $H$: a variable of $I_a$ in the goal reduced passes to a goal the reduction spawns, directly or by an assignment it adds for a reader that goal holds, by SRSW for the person process, whose clause passes the term of an ask to the relay it spawns, and by Definition~\ref{def:relay} for a relay, and any other person goal holds none; and a variable of $I_a$ the reduction introduces takes the place of a writer the goal held, and is read as its question.  A stutter.  But where the relay gives the place of a writer $X$ of $I_a$ that is not open in $c^i$ to another, that one is not in $I_a$ (Definition~\ref{def:sglp-goals}), so $H$ loses $X$, read as its question, and $G$ and $\sigma_r$ are unchanged: the Dismiss of $X$'s question by $a$.

(4) A Reduce at $a$ of a person goal that adds an assignment $X? := u$ with $X \in I_a$.  The goal is a relay, no other person goal holding a writer of $I_a$, and by Definition~\ref{def:relay} the assignment is the one answer of the reduction, $X$ a writer at a position the person writes that the relay received or that takes the place of one, read as its question, and $u$ is a term of the type of $X$, a list cell or \verb|[]| where $X$ is of a stream or list type, ground but for fresh variables: readers $Y_1?,\ldots,Y_k?$ at the positions the person writes of type \verb|Stream(|$T'$\verb|)| or \verb|List(|$T'$\verb|)|, but where the copy ends the list there with \verb|[]|, and writers $W_1,\ldots,W_m$ at the outermost positions the program writes, $Y_1,\ldots,Y_k$ and $W_1?,\ldots,W_m?$ passed on to relays of $a$.  If $X$ is not open in $c^i$ (Definition~\ref{def:sglp-goals}): $X \in H$ before, occurring in the relay reduced, and not after; the variables of $u$ are not in $I_a$ (Definition~\ref{def:sglp-goals}), so $H$ gains none; $\sigma_r$ does not take the assignment; $G$ is unchanged.  This is the Dismiss of $X$ by $a$.  Otherwise $X \in H$ before and not after; the variables of $u$ are in $I_a$, so $H' = H \setminus \{X\} \cup \{Y_1,\ldots,Y_k,W_1?,\ldots,W_m?\}$; $\sigma_r$ gains $X? := u$; and $G$ is unchanged.  This is the Answer of $a$ to $X$ with $u$ (Definition~\ref{def:rc-vmaglp}).

(5) A Communicate at $a$ delivering an ask, or a list of asks, to a \verb|merge| goal the compilation adds or to the person process, or delivering to a person goal an assignment other than an answer that a Reduce of a person goal added: the variables of $I_a$ it moves are in $H$ before and after, and $G$ and $\sigma_r$ are unchanged: a stutter.  A Communicate, within $a$ or from another agent, delivering to a relay of $a$ an assignment $X? := u$ that $\sigma_r$ holds, $X?$ held by the relay: $X?$ leaves $H$, and the variables of $u$, in $I_a$, and the assignment, shown, enter it; the sender's $\sigma_r$ loses the assignment; $G$ is unchanged.  This is the Communicate of the assignment to $a$'s person (Definition~\ref{def:rc-vmaglp}).  A Communicate of an assignment an asking clause's reduction made, or of one by which a writer takes another's place, replaces a variable by what $\sigma$ already read it as: a stutter.  A Communicate delivering to a platform goal an assignment added by a Reduce of a platform goal, within $a$ or from another agent, or one added by a relay to a variable of $I_a$: the source Communicate, $G$ and $\sigma_r$ changing as there and $H$ not at all.  A Communicate delivering to a person goal an assignment of a variable that is not in $I_a$: a stutter.

(6) A Cold-call from $a$ to $b$ moves a message between the network streams, which are arguments of platform goals: the source Cold-call.

\emph{Every Answer and every Dismiss is the image of a Reduce of a person goal}: by the cases, each is the image of a Reduce of a relay, of case (4) or, for a Dismiss, of the last sentence of case (3), and of no other.

\emph{Properness.}  Let a variable occur in the source goals of $a$ in $\sigma(c^{i+1})$ and not in $\sigma(c^i)$.  Read as in $G$, it occurs in a platform goal of $a$ in $c^{i+1}$ other than a \verb|merge| goal the compilation adds, and in no goal of $a$ in $c^i$: a variable of such a goal in $c^i$ is one of its source, read as in $G$; the variable of an ask reaches only \verb|merge| goals, the person process and a relay; and a variable a person goal holds reaches no platform goal, the variables of an answer being fresh and a writer that takes another's place being read as its question.  So it occurs in no goal of $a$ earlier, by the properness of the maGLP computation, hence in no $H$ earlier, the variables of $H$ occurring in goals of $a$ or in terms of $\sigma_a$ assigned to ask-stream variables or to readers that person goals of $a$ hold.  The clause each source Reduce renames is, but for its ask streams, the compiled clause that the Reduce of $r$ it is the image of renames apart from every agent's goals and readers substitution in $c^i$ and in every earlier configuration (Definitions~\ref{def:sglp-run} and~\ref{def:rc-proper}), so from every agent's source goals, held variables and readers substitution in $\sigma(c^i)$ and earlier.  Hence $\sigma(r)$ is proper (Definition~\ref{def:rc-proper}).
\end{proof}

\begin{corollaryrestated}
\runstatement
\end{corollaryrestated}
\begin{proof}
In $c^0$ agent $a$ holds the platform goal $g'_a$ and the person goal $\mathit{person}_{\pi(a)}(U_a?, s_a)$, which holds no variable of $I_a$, and $\sigma_a = \emptyset$, so $\sigma(c^0)$ has at each $a$ the agent state $(\emptyset, \{g_a\}, \emptyset)$, which is the initial local state of the vmaGLP transition system over $P$ and $M$ where $g_a$ is \texttt{agent(ch(\_?,\_),ch(\_?,\_))} (Definition~\ref{def:rc-vmaglp}); so $\sigma(c^0)$ is its initial configuration, and $\sigma(r)$, a proper safe computation from it by Theorem~\ref{thm:soundness}, is a proper safe run (Definition~\ref{def:rc-ts}).
\end{proof}

\section{Implementation Notes}\label{app:notes}

The decisions the implementation takes where Sections~\ref{sec:sglp}, \ref{sec:simulation} and~\ref{sec:implementation} leave a choice open, recorded here so that a reader of the code finds them stated.  The code is \verb|programs/sglp| \full{of the GLP repository}{in the supplementary material}: \verb|monitor.glp|; \verb|transform.glp|, the transformation, and \verb|transform.sh|, which runs it; and a directory per example with the vGLP program (\verb|graph.vglp|, \verb|coins.vglp|), its canonical compilation by hand, \verb|graph.glp| or \verb|coins.glp|, clause for clause as the README records, its sGLP source in the term form (\verb|graph_sglp.glp|, \verb|coins_sglp.glp|), the two modules the transformation prints from it --- \verb|profiles.glp| for the profiles and \verb|population.glp| for the population --- its harness, its run script and its reader of the log; and the tests, \verb|test_sglp.sh| and \verb|tests/|.

\begin{itemize}
\item The clock is in simulated seconds; a year is $365.25$ days, and the units of \verb|until| and of a rate are second, minute, hour, day, week and year, singular or plural after \verb|until|.  The horizon is inclusive: the clock never passes it, and the monitor stops when the least activation time would, binding every pending token to \verb|stop|, which ends the rated goal waiting on it, or when the queue is empty at rest, and writes \verb|clock(T)| last.
\item The monitor's queue is a skew heap on the pair of activation time and sequence number, so two requests with one time are released in the order queued.  A delay is $-\ln U$ times the unit's seconds over the rate's number, $U = K/(2^{31}-1)$, $K$ in $1\ldots2^{31}-2$ from \verb|random|, the monitor's seed threaded through its draws.
\item Seeds: the monitor draws from the run's seed; a second generator, started from the run seed's first draw, is cut into blocks of $4096$ draws, and from agent $a$'s block are drawn, in order, its profile per dimension and the seed of its person process, which draws for each person goal a value in $1\ldots2^{31}-2$ and hands it as the goal's seed, continuing with its own next seed: the goal's draws and the process's are on different seed chains, where handing the process's next seed itself put a goal's second draw and the next goal's first on one seed, GLP's \verb|random| making the next seed a function of the seed alone~\cite{shapiro2025glp}.  The harness reaches the next block by $4096$ draws made one after the other, not by a closed form, since GLP fixes only that the next seed is a function of the seed and not the generator.  The run's seed plus the agent's number was not used, since consecutive agents' first draws then step by a constant.  Each run's draw --- how many agents of each profile, and which agent has which --- is computed again from the seed, without the run, and written beside the log.
\item The agents are numbered $1,\ldots,n$, $n$ even and at least $4$, so that both lists of each agent of the social graph are non-empty, as \verb|Peers| requires; in the social graph the first half are of one sex and the rest of the other, and each agent's lists are in increasing order.  Placement is attribution: an agent's goals are its \verb|start| goal, its person process and what they spawn; the harness's goals belong to no agent.
\item The person process reads the ask stream of the agent's initial goal, in the dispatcher's place, and its answer and present procedures are the relay of each question (Definition~\ref{def:relay}): a person goal answers into its copy of the question, and the answer procedure hands the program's variable the answer once the positions the person writes are ground but for those of a stream or a list, which it hands over one element at a time, the outermost positions the program writes as fresh writers whose readers the person process keeps, a step after the person goal's; the program's term reaches the person goal the same way, by the present procedures, once the positions the program writes are ground, the person's positions inside it open; and a term crossing at any depth, a question inside an answer included, crosses the same way, so that every answer is logged.  Two positions the person writes in one alternative are two questions, each with its own answer procedure and log line, and a list the program writes inside a question is shown one element at a time, as an answer of a stream type is given; and the person process asks the monitor for the clock by a request \verb|time(T)|, answered at once, to log it.
\item The monitor's input and the log's sink are mutual references allocated by the root's \verb|allocate_mutual_reference/2| and written by \verb|stream_append/3|; the network is the root's \verb|mwm/2| over the agents' output streams and a balanced tree of \verb|deliver/4| goals.  The log's sink is read by \verb|send_to_user/1| and written tab-separated by the harness.
\item \verb|one_of| and the \verb|nth| it calls take a non-empty list, of the type \verb+NonEmpty(X) ::= [X | List(X)]+, and each has a clause for every list of that type, as the type checker requires of every procedure~\cite{shapiro2026types}.
\item The log is one line per answer: the simulated time in seconds, the agent, the interactive type and the term handed over, the positions the program writes as variables, separated by tabs.
\item The transformation prints one term a line, operators prefix, and chooses a question's person procedure by the functor of its term, so the interactive types of a program must have distinct functors, which it checks; each dimension's draw is printed with the population.  \verb|one_of| is shared by the two programs.
\item Coins among friends: the messages are \verb|pay|, \verb|coins|, \verb|propose| and \verb|decline|.  A wallet keeps its holdings in order, a new issuer's at the end, and a holding that reaches zero is removed; the order decides only which holding a request draws.  A pay of a friend whose coins the wallet does not hold would send zero coins and change nothing, and does not arise, since a request pays an issuer its menu shows.  While a card waits for its answer the agent takes neither a message nor a request; after a yes the wallet is rebuilt before either is taken, and a request answered while a card waited is acted on after the card; after a pay the agent asks again once the wallet after the pay is known.  The harness takes each agent's friends from the friendship graph of Section~\ref{sec:example}'s run of the same seed and refuses a graph on which an agent has no friend, an agent's friends being of the non-empty type \verb|Peers|.  The circulation and the holdings of Section~\ref{sec:currencies} are replayed from the log, each move at the answer the log shows for it; no move was deferred across a month's end in the runs reported.
\end{itemize}

\section{The Monitor: Gillespie's Algorithm in GLP}\label{app:monitor}

The monitor is the implementation of the core of Gillespie's algorithm~\cite{gillespie1977exact}, in the first reaction form of Section~\ref{sec:sglp}, written in GLP over the guard \verb|when_idle|.  This is the whole module as it is \full{in the repository}{in the supplementary material}, \verb|programs/sglp/monitor.glp|, its comments replaced by this text; the transformation, the two examples, their harnesses and their tests are in the same directory.

A request is \verb|rated(R, Tok)|, \verb|R| a rate \verb|N/U| and \verb|Tok| a token the monitor binds to \verb|go| at Release and to \verb|stop| when the run ends, or \verb|time(T)|, answered at once with the clock.  The queue is a skew heap: a node holds an activation time, a sequence number, the token and two subheaps.

{\footnotesize\begin{verbatim}
TimeUnit ::= second ; minute ; hour ; day ; week ; year.
Rate     ::= /(Number, TimeUnit).
Token    ::= go ; stop.
Request  ::= rated(Rate, Token?) ; time(Number?).
Queue    ::= nil ; node(Number, Integer, Token?, Queue, Queue).
Clock    ::= clock(Number).
\end{verbatim}}

The monitor goal holds its input stream, the clock in seconds, the queue, its seed, the next sequence number, the horizon in seconds and the log's reference, and starts with the clock at zero and the queue empty.

{\footnotesize\begin{verbatim}
exported procedure monitor(Stream(Request)?, Integer?, Number?,
                           MutualRef?).
monitor(In, Seed, Horizon, Log) :-
    monitor(In?, 0, nil, Seed?, 0, Horizon?, Log?).

procedure monitor(Stream(Request)?, Number?, Queue?, Integer?,
                  Integer?, Number?, MutualRef?).
monitor([rated(R, Tok?)|In], Clock, Q, Seed, Seq, Horizon, Log) :-
    random(Seed?, 2147483646, K, Seed1),
    delay(R?, K?, D),
    T := Clock? + D?,
    Seq1 := Seq? + 1,
    qmerge(node(T?, Seq?, Tok, nil, nil), Q?, Q1),
    monitor(In?, Clock?, Q1?, Seed1?, Seq1?, Horizon?, Log?).
monitor([time(Clock?)|In], Clock, Q, Seed, Seq, Horizon, Log) :-
    monitor(In?, Clock?, Q?, Seed?, Seq?, Horizon?, Log?).
monitor(In, _, node(T, _, go, L, R), Seed, Seq, Horizon, Log) :-
    when_idle, T? =< Horizon? |
    qmerge(L?, R?, Q),
    monitor(In?, T?, Q?, Seed?, Seq?, Horizon?, Log?).
monitor(_, Clock, node(T, S, K?, L, R), _, _, Horizon, Log) :-
    when_idle, T? > Horizon? |
    stop_all(node(T?, S?, K, L?, R?)),
    stream_append(clock(Clock?), Log?, _).
monitor(_, Clock, nil, _, _, _, Log) :-
    when_idle |
    stream_append(clock(Clock?), Log?, _).

procedure stop_all(Queue?).
stop_all(nil).
stop_all(node(_, _, stop, L, R)) :-
    stop_all(L?), stop_all(R?).
\end{verbatim}}

A rated request is assigned its activation time at once, the clock plus a delay drawn from the exponential distribution with its rate, and queued; this is the first reaction method, the time drawn when the goal is spawned.  A \verb|time| request is answered from the clock.  The three clauses guarded by \verb|when_idle| are taken only when the machine rests, so every request sent has arrived and the queue is that of the configuration: Release moves the clock to the least activation time, within the horizon, and binds that request's token to \verb|go|; the two stop clauses write the clock to the log when the least activation time is past the horizon or nothing is pending, the first after binding every pending token to \verb|stop| by \verb|stop_all|, which ends each rated goal still waiting.  The clauses are tried in order, so a request that has arrived is queued before the machine is found idle.

{\footnotesize\begin{verbatim}
procedure delay(Rate?, Integer?, Number).
delay(N/U, K, D?) :-
    number(N?), N? > 0 |
    seconds(U?, S),
    P := K? / 2147483647,
    L := ln(P?),
    D := neg(L?) * S? / N?.

procedure seconds(TimeUnit?, Integer).
seconds(second, 1).
seconds(minute, 60).
seconds(hour, 3600).
seconds(day, 86400).
seconds(week, 604800).
seconds(year, 31557600).
\end{verbatim}}

The delay of a rate of \verb|N| per unit \verb|U| is $-\ln(K/(2^{31}-1))\cdot S/N$ seconds, $S$ the seconds in \verb|U| and $K$ uniform in $1\ldots2^{31}-2$ from the monitor's seed, which is threaded through its draws; a year is $365.25$ days.

{\footnotesize\begin{verbatim}
procedure qmerge(Queue?, Queue?, Queue).
qmerge(nil, Q, Q?).
qmerge(node(T, S, K?, L, R), nil, node(T?, S?, K, L?, R?)).
qmerge(node(T1, S1, K1?, L1, R1), node(T2, S2, K2?, L2, R2),
       node(T1?, S1?, K1, M?, L1?)) :-
    T1? < T2? |
    qmerge(R1?, node(T2?, S2?, K2, L2?, R2?), M).
qmerge(node(T1, S1, K1?, L1, R1), node(T2, S2, K2?, L2, R2),
       node(T1?, S1?, K1, M?, L1?)) :-
    T1? =:= T2?, S1? < S2? |
    qmerge(R1?, node(T2?, S2?, K2, L2?, R2?), M).
qmerge(node(T1, S1, K1?, L1, R1), node(T2, S2, K2?, L2, R2),
       node(T2?, S2?, K2, M?, L2?)) :-
    T2? < T1? |
    qmerge(R2?, node(T1?, S1?, K1, L1?, R1?), M).
qmerge(node(T1, S1, K1?, L1, R1), node(T2, S2, K2?, L2, R2),
       node(T2?, S2?, K2, M?, L2?)) :-
    T2? =:= T1?, S2? < S1? |
    qmerge(R2?, node(T1?, S1?, K1, L1?, R1?), M).
\end{verbatim}}

\verb|qmerge| merges two skew heaps, the least pair of activation time and sequence number at the root, so that two requests with one activation time are released in the order they were queued; the token writer of each node is passed on with the node.

\section{Related Work, Continued}\label{app:related}

Section~\ref{sec:related-work}'s six lines of work, each expanded.

\mypara{Calculi and languages with rates} BioSPI implemented the biochemical stochastic pi-calculus in Flat Concurrent Prolog~\cite{shapiro1989family}: each channel is an object carrying a base rate, processes post send, receive, dimer and withdraw requests to it, a monitor and the channel objects jointly compute the weighted rates, and the monitor draws the channel, completes one transmission, withdraws the losing offers of the two chosen processes and advances a clock~\cite{regev2000representation,priami2001application,silverman1988logix}.  Bio-PEPA~\cite{ciocchetta2009biopepa} is a process algebra for biochemical networks in the PEPA family, with stochastic simulation by Gillespie's method among its analyses; stochastic concurrent constraint programming~\cite{bortolussi2008stochastic} attaches a rate to each action of a concurrent constraint program, with a continuous-time Markov chain semantics simulated by Gillespie's method, which makes it the nearest language to sGLP: a concurrent constraint language with rates.  The difference is where the rate sits and what is simulated: in sCCP every action is rated and the program is the model, in sGLP the platform's program keeps vGLP's semantics unchanged and only the person's goals are rated.

\mypara{Two criticisms of BioSPI's implementation} The stochastic pi-machine makes the cost of a simulation depend on the number of species rather than the number of molecules, by grouping identical processes and counting the copies, and reports its gene example running in ten seconds where their own previous version of SPiM took eight minutes and BioSPI twenty-one; the criticism is that the earlier simulators, BioSPI among them, describe no algorithm for keeping track of identical processes, the separate process per molecule being what they report of their previous version~\cite{phillips2007efficient}.  A later machine for stochastic process calculi carries the same point to restricted channels, and against the stochastic pi-machine itself: it stores each complex as one species, where a simulator in which every restricted channel has a globally fresh name carries one species and one reaction per copy~\cite{pauleve2010generic}.  The second criticism is of a different kind: the stochastic pi-machine is proved reduction-equivalent to the calculus and shown to preserve apparent rates, which they report as the first correctness proof for the stochastic pi-calculus to cover optimisations of this kind; their own earlier machine was proved correct against a variant of the calculus~\cite{phillips2007efficient}.

\mypara{What sGLP does about them} The first economy is not available to us, and not because of how we implement: two agents are interchangeable only while their machine states agree, and in a grassroots platform they stop agreeing as soon as they have different friends, so a population of a thousand people is a thousand states, no two with the same open questions.  What carries over is the rest of the machinery --- the next of $M$ pending goals is found in $O(\log M)$ by a priority queue, where the linear scan of the direct method costs $O(M)$ and a tree of propensities $O(\log M)$~\cite{masuda2022gillespie,koster2024fast} --- together with one thing a vGLP program has for nothing: a reduction changes the resolvent of one agent, so an event touches the questions of that agent and of no other.  Their grouping also needs every choice to carry an identifier, and where a choice has none their encoding creates an extra definition for it; a person declaration names an interactive type, so we have the identifier already~\cite{phillips2007efficient}.  They also record a restricted channel's apparent rate as a delay, to keep the machine compositional where several machines run in parallel, on distributed or multi-core systems~\cite{phillips2007efficient}; that is the setting of a deployed platform, one machine per agent, and not yet of our simulator, which hosts the population in one machine (Section~\ref{sec:implementation}).  The second criticism we accept as the standard to meet, and meet it in the same form: the runs of the simulator are runs of the vmaGLP transition system of the very program being simulated, with the person's answers drawn rather than given.

\mypara{Rule-based languages} Kappa's simulator~\cite{danos2007scalable} keeps a superset of the events that may happen next, rejects the draws that do not correspond to a rule application, corrects the time advance so that the timings stay exact, and wakes only the rules a static activation relation names; its cost per event is independent of the number of species and of the number of agents, and logarithmic in the number of rules.  In sGLP the rules are the interactive types, of which a platform has few, and the instances are the questions open at the agents, of which a population has many; the activation relation it computes statically we have from the semantics, since a reduction reaches one agent.

\mypara{Continuous-time agent-based simulation} In ML3~\cite{koster2024fast} agents cannot be grouped into subpopulations, so handling the transitions efficiently is what decides the cost, and the dependency graph is dynamic where a biochemical one is static.  DEVS~\cite{zeigler2000theory} is the formalism of discrete-event simulation --- atomic models with time advances, coupled into networks --- on which ML3 and the discrete-event simulators of agents rest; sGLP's Release is a discrete event in that sense, and its machine steps between events take no time.

\mypara{Agent-based platforms, and testing in simulation} In NetLogo, Repast, MASON, Mesa and krABMaga~\cite{antelmi2024reliable}, and in agent-based modelling at large, the clock is a discrete step: the choice of step size bears on the validity of the results, and an alternative that is retracted when another decision wins --- a stochastic race --- has no formulation~\cite{koster2024fast}.  A platform's agent is the whole of what is modelled, where a vGLP agent is a person and a machine, and only the person is drawn.  Baiardi, Burattini, Ciatto and Pianini~\cite{baiardi2026testing} argue that a multiagent system should be tested in simulation on the very specification that is deployed, with no surrogate, and map the control flow of BDI agents onto a discrete-event simulator at several granularities; the demand is ours too, and what is tested differs: their agents are the system under test and their environment is the simulator's, where in sGLP the system under test is a vGLP program whose agents are people and machines, and the simulation replaces the people.

\mypara{Simulated people with profiles} Generative agents~\cite{park2023generative} are simulated people whose memory, reflection and planning are done by a language model from a persona written as text, in a sandbox town; OASIS~\cite{yang2024oasis} scales such agents to a million on a simulated social platform, with recommendation and the platform's actions modelled; user simulation in information retrieval~\cite{balog2024user} evaluates search and recommendation systems by simulated users, and PersonaEval~\cite{liu2026personaeval} evaluates interactive applications by persona-based simulated users driven by language models.  All simulate the person by a persona and a model; none runs the deployed program, and none specifies the person's behaviour stochastically, with a stated law.  An sGLP profile is the person's behaviour as a program with rates, typed by the dual of the question it answers, so that the simulation is a run of the deployed program (Theorem~\ref{thm:soundness}), drawn from a seed and reproducible; it lacks the persona's richness, and the rates and the profiles are to be estimated from how people use the deployed platforms.

\mypara{What the examples study} McPherson, Smith-Lovin and Cook review homophily across the dimensions of social life, sex, age, race, religion, education and the rest, and find it structuring every kind of tie~\cite{mcpherson2001birds}.  In Currarini, Jackson and Pin's model agents of several types form friendships by meeting, with a preference bias for their own type, and they derive how the share of same-type friendships and the number of friends vary with a group's relative size and its bias, which they fit to the Add Health data~\cite{currarini2009economic}; sGLP's homophile and indifferent profiles are the bias, and Table~\ref{tab:social-graph-mix} and Figure~\ref{fig:social-graph-matrix} are what the deployed platform does under it, not a model of it.  On mutual credit, Iosifidis et al.\ analyse the cyclic motifs of the Sardex transaction network, where credit returns to its issuer around cycles~\cite{iosifidis2018sardex}; Fleischman, Dini and Littera measure the liquidity obligation-clearing and mutual credit save on that network's data~\cite{fleischman2020liquidity}; Littera et al.\ describe the system's design and history~\cite{littera2017sardex}; and Schraven models why a LETS need not collapse under opportunism~\cite{schraven2001mutual}.  None simulates the deployed program; Section~\ref{sec:currencies}'s runs are of coins among friends as the platform's program executes it, with the people drawn.

\mypara{The exponential assumption} The non-Markovian Gillespie algorithm simulates coupled renewal processes by a first-order expansion accurate for large numbers of processes, and the Laplace Gillespie algorithm simulates exactly those whose inter-event distribution is a mixture of exponentials, which covers the fat-tailed distributions common in human interaction~\cite{masuda2022gillespie}; a rate per act is a modelling assumption and not a fact about people.

\section{Definitions and Results Recalled from vGLP}\label{app:recalled}

Sections~\ref{sec:sglp} and~\ref{sec:simulation} and Theorem~\ref{thm:soundness} rest on the following definitions and results of vGLP~\cite{shapiro2026volition}, which rest in turn on GLP~\cite{shapiro2025glp} and its types~\cite{shapiro2026types}; they are recalled as stated there, each named by its title in that paper.  Reduction, the writer mgu, goal identity, the transaction closure and the transactions-based multiagent transition system, which the maGLP and vmaGLP transition systems are defined over, and the construct process and the dispatcher, which the canonical compilation names, are those of~\cite{shapiro2025glp,shapiro2026volition} and are not recalled; a recalled definition names each by its title in that paper.

\begin{definition}[Transition System, Computation, Run, Safe, Live, Correct~\cite{shapiro2021multiagent,lewis2026volitional,shapiro2026volition}]\label{def:rc-ts}
A \temph{transition system} is a tuple $TS = (C, c_0, T, {\sim})$ where $C$ is an arbitrary set of \temph{configurations}, $c_0 \in C$ a designated \temph{initial configuration}, $T \subseteq C \times C$ a set of \temph{transitions}, each a pair $c \rightarrow c'$ of non-identical configurations $c \ne c' \in C$, and ${\sim}$ a partial equivalence relation on $T$ --- an equivalence on a subset $\mathit{dom}({\sim}) \subseteq T$; we write $[t]$ for the class of $t \in \mathit{dom}({\sim})$ under ${\sim}$.

A \temph{computation} is a (nonempty, finite or infinite) sequence of configurations $c_1, c_2, \ldots$; it is a \temph{run} if $c_1 = c_0$, and \temph{safe} if $c_i \rightarrow c_{i+1} \in T$ for every two consecutive configurations. We write $c \xrightarrow{*} c'$ for the existence of a safe computation from $c$ to $c'$ (empty if $c = c'$). A class $[t]$ is \temph{enabled} in $c$ if $c \rightarrow c' \in [t]$ for some $c'$. A run is \temph{live} if no class $[t]$, $t \in \mathit{dom}({\sim})$, is enabled in every configuration of some suffix in which no member of $[t]$ occurs, and \temph{correct} if it is safe and live.
\end{definition}

\begin{definition}[Implementation (Refinement)~\cite{shapiro2021multiagent,shapiro2026implementing,lewis2026volitional,shapiro2026volition}]\label{def:rc-impl}
Given two transition systems $TS = (C, c_0, T, {\sim})$, the \temph{specification}, and $TS' = (C', c'_0, T', {\sim'})$, an \temph{implementation of $TS$ by $TS'$} is a function $\sigma : C' \rightarrow C$ with $\sigma(c'_0) = c_0$, in which case the pair $(TS', \sigma)$ is \temph{an implementation of $TS$}.  Given a computation $r' = c'_1, c'_2, \ldots$ of $TS'$, $\sigma(r')$ is the computation $\sigma(c'_1), \sigma(c'_2), \ldots$ with consecutive repetitions removed; a transition $c' \rightarrow d'$ of $TS'$ with $\sigma(c') = \sigma(d')$ is a \temph{stutter}, contributing no step to $\sigma(r')$.
\end{definition}

\begin{definition}[Variables, Terms, Goals, Clauses~\cite{shapiro2025glp,shapiro2026volition}]\label{def:rc-variables}
Let $\calV$ denote the set of LP variables (identifiers beginning with uppercase), henceforth called \temph{writers}, and $\calV? = \{X? \mid X \in \calV\}$ the \temph{readers}; $\hat\calV = \calV \cup \calV?$, and a writer $X$ and its reader $X?$ form a \temph{variable pair}.  Terms, unit goals, goals, and clauses are as in LP, over $\hat\calV$.  A term, goal, or clause satisfies the \temph{single-occurrence (SO) invariant} if every variable occurs in it at most once, and a clause satisfies the \temph{single-reader/single-writer (SRSW) restriction} if it satisfies SO and a variable occurs in it iff its paired variable does.
\end{definition}

\begin{definition}[Guarded Clause, Procedure, GLP Program~\cite{shapiro2025glp,shapiro2026volition}]\label{def:rc-glp-program}
A \temph{guarded clause} has the form $H$ \verb|:-| $G$ \verb"|" $B$, where $H$ is the head, $G$ a conjunction of guard predicates, and $B$ the body; guard arguments are readers paired to head writers.  A \temph{GLP program} is a finite sequence of guarded clauses, each satisfying SRSW; the clauses for the same predicate form a \temph{procedure}.  The set of goals $\hat\calG(M)$ of a program $M$ includes all goals over $\hat\calV$ and the vocabulary of $M$ that satisfy SO.
\end{definition}

\begin{definition}[maGLP Transition System~\cite{shapiro2025glp,shapiro2026volition}]\label{def:rc-glp}
Given a GLP program $M$, an \temph{asynchronous resolvent} over $M$ is a pair $(G, \sigma)$ where $G \in \hat\calG(M)$ and $\sigma$ is a readers substitution.  Given agents $P\subset \Pi$, the \temph{maGLP transition system} over $P$ and $M$ is the transactions-based multiagent transition system (Definition ``Transactions-Based Multiagent Transition System'' of~\cite{shapiro2026volition}) over $P$, local states being asynchronous resolvents, initial local state $(\{\verb|agent(ch(_?,_),ch(_?,_))|\}, \emptyset)$, the transactions given by the rules below, and the equivalence ${\sim}$ below.  Each rule is read as in an operational semantics: above the line its condition, below it the transaction, from its participants' local states to their states after it, every other agent stationary; the prose before each rule says what the premise names.
\begin{description}
\item[Reduce $p$.]  $A \in G_p$ is a unit goal, and $C \in M$ is the first clause for which the reduction of $A$ with $C$ (Definition ``Reduction'' of~\cite{shapiro2026volition}) succeeds, with result $(B, \hat\sigma)$.
\[
\frac{A \in G_p \qquad C \text{ the first clause reducing } A, \text{ with result } (B, \hat\sigma)}
     {(G_p,\, \sigma_p) \;\rightarrow\; ((G_p \setminus \{A\} \cup B)\hat\sigma,\; \sigma_p \circ \hat\sigma?)}
\]
\item[Communicate $p$ to $q$.]  Participants $p, q \in P$, possibly $p = q$: $\{X? := T\} \in \sigma_p$ and $X?$ occurs in $G_q$.
\[
\frac{\{X? := T\} \in \sigma_p \qquad X? \text{ occurs in } G_q}
     {(G_p,\, \sigma_p),\ (G_q,\, \sigma_q) \;\rightarrow\; (G_p,\, \sigma_p \setminus \{X? := T\}),\ (G_q\{X? := T\},\, \sigma_q)}
\]
\item[Cold-call $p$ to $q$.]  Participants $p \ne q \in P$: the network output stream in $(G_p, \sigma_p)$ has a new message \verb|msg|$(q,X)$; $(G'_p, \sigma'_p)$ is the result of advancing the network output stream and $(G'_q, \sigma'_q)$ the result of adding \verb|msg|$(q,X)$ to the network input stream.
\[
\frac{\text{the network output stream of } p \text{ has a new } \mathtt{msg}(q,X)}
     {(G_p,\, \sigma_p),\ (G_q,\, \sigma_q) \;\rightarrow\; (G'_p,\, \sigma'_p),\ (G'_q,\, \sigma'_q)}
\]
\end{description}
${\sim}$ is the equivalence relating $t_1 \sim t_2$ iff both are Reduce transactions at the same agent $p$ reducing the same goal (by identity, Definition ``Goal Identity'' of~\cite{shapiro2026volition}) with the same clause, both are Communicate transactions from the same $p$ to the same $q$ applying the counterpart of the same writer assignment to the same goal, or both are Cold-call transactions from the same $p$ to the same $q$ delivering the same message.
\end{definition}

\begin{definition}[Volitional Procedure, Interactive Type, Interactive Term, Ordinary Clause, vGLP Program~\cite{shapiro2026volition}]\label{def:rc-program}
A \temph{volitional procedure} $p$ of arity $n$ is declared \verb|procedure (|$T$\verb|)*|$p(T_1,\ldots,T_n)$\verb|.|, where $T_1,\ldots,T_n$ are the types of its arguments and $T$, in writer or reader mode as an argument type is, is its \temph{interactive type}.  An interactive type is \temph{flat}: it is built from constants, \verb|String|, \verb|Integer|, \verb|Real|, \verb|Key|, \verb|Constant|, \verb|Module|, named types so built, tuples and unions, and \verb|Stream| or \verb|List| as the whole type or as an argument of a tuple, never under another \verb|Stream| or \verb|List|; it has no other recursive type, and no type parameter at a position the person writes.  A clause of it has the form \verb|(|$A$\verb|)*|$p(S_1,\ldots,S_n)$ \verb|:-| $G$ \verb"|" $B$, where the \temph{interactive term} $A$ is a term of type $T$, possibly a variable but not the anonymous variable, and it is the guarded clause $p(S_1,\ldots,S_n,A)$ \verb|:-| $G$ \verb"|" $B$ (Definition~\ref{def:rc-glp-program}), of arity $n{+}1$, whose last argument is matched against the goal's end of its interactive variable (Definition~\ref{def:rc-vmaglp}).  A clause of a volitional procedure reads an argument at an input position, by a non-variable term at that position in its head or by a guard on it, only if every later clause of the procedure reads it too, and every call of a volitional procedure, in a clause body or in the initial goal, holds a writer at each output position.  A clause is \temph{ordinary} if it is not a clause of a volitional procedure.  A \temph{vGLP program} is a finite sequence of ordinary clauses and clauses of volitional procedures, each satisfying SRSW, a clause of a volitional procedure satisfying it as a guarded clause of arity $n{+}1$; the clauses for the same predicate form a \temph{procedure}.  The set of goals $\hat\calG(M)$ of a vGLP program $M$ includes all goals over $\hat\calV$ and the vocabulary of $M$ that satisfy SO, a goal of a volitional procedure of arity $n$ being $n$-ary until it is asked and of arity $n{+}1$ after (Definition~\ref{def:rc-vmaglp}).
\end{definition}

\begin{definition}[vmaGLP Transition System~\cite{shapiro2026volition}]\label{def:rc-vmaglp}
Given a vGLP program $M$, an \temph{agent state} is a triple $(H_p, G_p, \sigma_p)$ of a finite set $H_p$ of typed variables \temph{held} by the agent's person and of assignments \temph{shown} to them and an asynchronous resolvent $(G_p, \sigma_p)$ over $M$ (Definition~\ref{def:rc-glp}).  Given agents $P\subset \Pi$, the \temph{vmaGLP transition system} over $P$ and $M$ is the transactions-based multiagent transition system (Definition ``Transactions-Based Multiagent Transition System'' of~\cite{shapiro2026volition}) over $P$, local states being agent states, initial local state $(\emptyset, \{\verb|agent(ch(_?,_),ch(_?,_))|\}, \emptyset)$, the transactions being the Reduce, Communicate and Cold-call transactions of maGLP and the three rules below, read as those of Definition~\ref{def:rc-glp}, and the equivalence ${\sim}$ below.  Reduce applies to a unit goal that is not an $n$-ary goal of a volitional procedure, a clause of a volitional procedure taken as its guarded clause (Definition~\ref{def:rc-program}).  Communicate delivers an assignment $\{X? := T\}\in\sigma_p$ to a reader $X?\in H_q$ as to one occurring in $G_q$, then removing $X?$ from $H_q$ and adding the variables of $T$, each with the type of its position in $T$, and the assignment, shown; Reduce, Cold-call and every other Communicate leave the held variables unchanged.
\begin{description}
\item[Ask $p$.]  $A$ is an $n$-ary goal of a volitional procedure $q$ with interactive type $T$, $X$ a writer occurring nowhere in the configuration, the \temph{interactive variable} of the asked goal, of type $T'$ where $T$ is $T'$ or $T'?$; the goal takes its end of $X$ and the person the other: $(A', Y) = (q(S_1,\ldots,S_n,X?), X)$ where $T$ is in reader mode, and $(q(S_1,\ldots,S_n,X), X?)$ where it is in writer mode.  The asked goal keeps the identifier of $A$.
\[
\frac{A = q(S_1,\ldots,S_n) \in G_p \qquad X \text{ fresh} \qquad (A', Y) \text{ as above}}
     {(H_p,\, G_p,\, \sigma_p) \;\rightarrow\; (H_p \cup \{Y\},\; G_p \setminus \{A\} \cup \{A'\},\; \sigma_p)}
\]
\item[Answer $p$.]  $X \in H_p$ is a writer whose reader $X?$ occurs in the configuration, and $t$ a term of the type of $X$, ground but for variables occurring nowhere in the configuration: a reader $Y_i?$ at each position the person writes of type \verb|Stream(|$T'$\verb|)| or \verb|List(|$T'$\verb|)|, unless the person ends the list there with \verb|[]|, and a writer $W_j$ at each outermost position the program writes; where $X$ itself is of such a type, $t$ is \verb|[|$e$\verb#|#$Y_1?$\verb|]| with $e$ such a term of type $T'$, or \verb|[]| for a list.  The person holds $Y_i$ and $W_j?$, each with the type of its position in $t$.
\[
\frac{X \in H_p \qquad X? \text{ occurs} \qquad t \text{ as above}}
     {(H_p,\, G_p,\, \sigma_p) \;\rightarrow\; (H_p \setminus \{X\} \cup \{Y_1,\ldots,Y_k,W_1?,\ldots,W_m?\},\; G_p,\; \sigma_p \cup \{X? := t\})}
\]
\item[Dismiss $p$.]  $X \in H_p$ is a writer whose reader $X?$ occurs nowhere in the configuration.
\[
\frac{X \in H_p \qquad X? \text{ occurs nowhere}}
     {(H_p,\, G_p,\, \sigma_p) \;\rightarrow\; (H_p \setminus \{X\},\; G_p,\; \sigma_p)}
\]
\end{description}
${\sim}$ is the partial equivalence whose domain is the Ask, Reduce, Communicate and Cold-call transactions, relating two Reduce, Communicate or Cold-call transactions as maGLP does, the person counting as a goal, and $t_1 \sim t_2$ iff both are Ask transactions at the same agent $p$ asking the same goal (by identity, Definition ``Goal Identity'' of~\cite{shapiro2026volition}).
\end{definition}

\begin{definition}[Proper Run~\cite{shapiro2026volition}, adapted from~\cite{shapiro2025glp}]\label{def:rc-proper}
A vmaGLP run is \temph{proper} if, at each agent $p$, a variable that occurs in $G_p$ after a transition but not before it also does not occur in $G_p$ or in $H_p$ at any earlier configuration, and the clause each Reduce renames is renamed apart from every agent's goals, held variables and readers substitution, in the configuration it reduces and in every earlier one.
\end{definition}

\begin{definition}[Question, Open, Closed, Volition~\cite{shapiro2026volition}]\label{def:rc-question}
In a configuration of the vmaGLP transition system, a writer $X \in H_p$ is a \temph{question} of $p$'s person: an interactive variable the person writes, introduced by an Ask, a writer inside the program's output delivered to the person, or the writer of a stream or list that an Answer leaves to the person, the tail of the stream or list answered or one inside the answer; it is \temph{open} if its reader $X?$ occurs in the configuration, and \temph{closed} otherwise.  A \temph{volition} of $p$'s person is an Answer transition of $p$: the person's answer to a question.
\end{definition}

\begin{theorem}[Volitional Soundness~\cite{shapiro2026volition}]\label{thm:rc-vs}
Let $c^1 \rightarrow c^2 \rightarrow \cdots$ be a proper safe run of the vmaGLP transition system over $P$ and $M$, $p \in P$, and $X$ a question of $p$'s person in $c^i$.  Then the writer $X$ occurs in $c^i$ in $H_p$ and nowhere else, and a transition $c^i \rightarrow c^{i+1}$ that adds an assignment $X? := t$ to a readers substitution is the Answer of $p$ to $X$.
\end{theorem}

The proof is in~\cite{shapiro2026volition}.

\begin{definition}[Person Channel, Person Writer, GLP with Persons, Grant~\cite{shapiro2026volition}]\label{def:rc-persons}
Let $N$ be a GLP program and $P\subset\Pi$.  A \temph{person channel} of a goal $G_0$ satisfying SO is a variable pair $(U, U?)$ such that $U?$ occurs in $G_0$ and $U$ does not; $U$ is its \temph{person writer}.  \temph{GLP with persons} over $P$ and $N$, with a designated person channel of each agent's initial resolvent, is the transition system obtained from the maGLP transition system over $P$ and $N$ (Definition~\ref{def:rc-glp}) by adding, at each agent $p$, \temph{grant} transitions
\[
(G_p, \sigma_p) \;\rightarrow\; (G_p,\; \sigma_p \circ \{U? := [c\,|\,U_1?]\})
\]
where $c$ is a ground term \verb|input(|$\mathit{Id}$\verb|, |$P$\verb|, |$R$\verb|)|, $U$ the person writer of $p$, and $U_1$ a fresh variable, the person writer after the transition, every other agent stationary; the partial equivalence is the maGLP transaction equivalence, with grant transitions outside its domain.
\end{definition}

\begin{definition}[Canonical Compilation~\cite{shapiro2026volition}]\label{def:rc-compilation}
The \temph{canonical compilation} $\lceil M\rceil$ of a vGLP program $M$ is the GLP program defined as follows.  A procedure of $M$ \temph{reaches a question} if it is volitional or a clause of it calls a procedure that does; each such procedure has one more argument, last, its \temph{ask stream}, which its goal writes and the dispatcher reads.  In a clause of such a procedure, each call in the body of a procedure that reaches a question is given a fresh writer as its ask stream, and the head carries $D?$, where $D$ is the clause's ask stream, with the body's ask streams merged into $D$: the call is given $D$ itself where there is one such call, \verb|merge| goals merge them where there are more, and the head carries \verb|[]| in place of $D?$ where there is none.  Then $\lceil M\rceil$ consists of: every ordinary clause of $M$, so extended; for each volitional procedure $q$ of $M$ with interactive type $T$, its clauses as guarded clauses of arity $n{+}1$ under a fresh name, \verb|q1| below, extended, and the asking clause, under the name $q$ so that every call of $q$ reaches it unchanged,
\begin{verbatim}
q(S1, ..., Sn, [ask(T, t(X))|D?]) :- q1(S1', ..., Sn', X?, D).
\end{verbatim}
where $S'_l$ passes the argument on --- at an input position the head carries $S_l$ and $S'_l$ is its reader $S_l?$, and at an output position the head carries $S_l?$ and $S'_l$ is the writer $S_l$ --- \verb|X| and \verb|X?| are exchanged where $T$ is in writer mode, and \verb|t| is the name of $T$ in lowercase, its type arguments appended, with \verb|_w| for writer mode and \verb|_r| for reader mode --- \verb|card_w|, \verb|stream_request_r| --- so that the asks of $M$ are of one type, the union over those functors; for each interactive type $T$ of $M$, the \temph{construct process} of $T$, which the dispatcher spawns on an ask with its variable and its channels to the dispatcher, and which acts as above (Definition ``Construct Process'' of~\cite{shapiro2026volition}); and the \temph{dispatcher}, which reads the ask stream of the initial goal and the person channel, spawns a construct process with a fresh identifier on each ask, and routes the person's grants to the construct processes by identifier, by the standard stream techniques (\cite{shapiro2026volition} prints it).  The initial goal of $\lceil M\rceil$ at each agent is that of $M$, extended, and the dispatcher on the ask stream of that goal and on the person channel, which is its person side.  No vGLP module imports from another module a procedure that reaches a question, so the questions of a program are those of one module; the GLP module that forms the program's initial goal imports that module's compiled procedures, ask streams included, and the dispatcher.  A \texttt{.vglp} source in the directory tree of a program is compiled into a module of that program, as a \texttt{.glp} source is.  The compilation is typed: $\lceil M\rceil$ carries the type definitions and procedure declarations of $M$, each procedure that reaches a question declared with its ask stream added, a volitional procedure's declaration being its asking clause's and its guarded clauses declared under their fresh name with the argument of its interactive type added in the mode of that type, the type of the asks, and the declarations of the construct processes and the dispatcher.  A vGLP program declares the type of its output to its person by defining the type \verb|Output|, a reserved name; the dispatcher draws that output as a thread whose elements are drawn by the widget of \verb|Output|.
\end{definition}

\begin{definition}[Compilation Mapping~\cite{shapiro2026volition}]\label{def:rc-sigma}
Given a configuration of GLP with persons over $P$ and $\lceil M\rceil$, the mapping $\sigma$ yields the vmaGLP configuration over $P$ and $M$ with, at each agent:
\begin{itemize}
\item $G$ = the multiset of the agent's goals, each read back to its source with its ask stream dropped: $q(S_1,\ldots,S_n,D)$ as the $n$-ary goal $q(S_1,\ldots,S_n)$, $q_1(S_1,\ldots,S_n,X?,D)$ and $q_1(S_1,\ldots,S_n,X,D)$ as the asked goal, with the identifier of the goal its asking clause reduced, and any other compiled goal as itself; the construct processes, the dispatcher and the \verb|merge| goals the compilation adds contributing nothing; and in every goal an asking clause's pair for an argument read as that argument: at an input position $S_l?$ as the term $u$, until the assignment is communicated, and at an output position $S_l$ as $W$ and $S_l?$ as $W?$;
\item $H$ = the variables held for the person, by the construct processes and by the asks the dispatcher has not yet taken, each with its type: the writers of the positions the person writes and the readers of the positions the program writes; and the assignments communicated to the construct processes, shown;
\item $\sigma_r$ = the restriction of the accumulated readers substitution to the readers of \temph{source variables} --- the variables of the initial goal, of the compiled clauses of $M$ as renamed, but for their ask streams, the interactive variables, and the variables of the terms construct processes bind questions with --- less the assignments the asking clauses' reductions make, but for the assignment $W? := S_l?$ of an output position, which it reads through --- absent while $S_l?$ is unassigned, and $W? := u$ once $S_l? := u$ is --- and less the assignments by which construct processes bind closed questions.
\end{itemize}
\end{definition}

\begin{lemma}[Compilation, per Transition~\cite{shapiro2026volition}]\label{lem:rc-correspondence}
Let $\sigma$ be the mapping of Definition~\ref{def:rc-sigma} and $c \rightarrow c'$ a transition of GLP with persons over $P$ and $\lceil M\rceil$ at agent $p$.  A grant, a reduction of the dispatcher or of a \verb|merge| goal the compilation adds, a reduction of a construct process that binds nothing, and a Communicate on the person channel, on an ask stream, or between the dispatcher and the construct processes are stutters.  A construct process binding a question whose reader occurs nowhere maps to the Dismiss of it by $p$.  A reduction of the asking clause of $q$ maps to the Ask of the source goal.  A Communicate of an assignment an asking clause's reduction made is a stutter.  A construct process binding an open question $X$ with the term $t$ its person's input forms, a term of the type of $X$ ground at the positions the person writes but for the stream and list positions, which hold readers of new writers it holds or, for a list the person ends, \verb|[]|, with new writers at the positions the program writes, whose readers it holds, maps to the Answer of $p$ to $X$ with $t$.  A reduction of a compiled clause of $M$ maps to the Reduce of the source goal with the source clause.  A Communicate, within the agent or from another, maps to the source Communicate of the same assignment, to the person where its reader is held for the person, and a Cold-call to the source Cold-call.
\end{lemma}

The lemma is the safety case of the proof of the Compilation theorem of~\cite{shapiro2026volition}, stated transition by transition.

\end{document}